\documentclass[]{interact}
\usepackage{epstopdf}
\usepackage[caption=false]{subfig}

\usepackage[longnamesfirst,sort]{natbib}
\bibpunct[, ]{(}{)}{;}{a}{,}{,}
\usepackage{natbib}

\usepackage{amsmath}
\usepackage{amssymb}
\usepackage{algorithmic}
\usepackage{algorithm}
\usepackage{caption}
\usepackage{mathtools}
\usepackage{changes}
\usepackage{multirow}
\usepackage{verbatim}
\usepackage{mathrsfs}
\usepackage{graphicx}
\usepackage{amsmath,amsthm,bm,mathrsfs}

\theoremstyle{plain}
\newtheorem{theorem}{Theorem}[section]

\newtheorem{proposition}[theorem]{Proposition}

\theoremstyle{definition}

\theoremstyle{remark}
\newtheorem{remark}{Remark}

\begin{document}

\title{On frequency- and time-limited $\mathcal{H}_2$-optimal model order reduction}

\author{
\name{Umair~Zulfiqar\textsuperscript{a}\thanks{CONTACT Umair~Zulfiqar. Email: umair.zulfiqar@research.uwa.edu.au}, Victor Sreeram\textsuperscript{a}, and Xin~Du\textsuperscript{b,c,d}}
\affil{\textsuperscript{a}School of Electrical, Electronics and Computer Engineering, The University of Western Australia (UWA), Perth, Australia; \textsuperscript{b}School of Mechatronic Engineering and Automation, and Shanghai Key Laboratory of Power Station Automation Technology, Shanghai University, Shanghai, China; \textsuperscript{c}Key Laboratory of Knowledge Automation for Industrial Processes, Ministry of Education, Beijing, China; \textsuperscript{d}Key Laboratory of Modern Power System Simulation and Control \& Renewable Energy Technology, Ministry of Education (Northeast Electric Power University), Jilin, China}
}

\maketitle

\begin{abstract}
In this paper, the problems of frequency-limited and time-limited $\mathcal{H}_2$-optimal model order reduction of linear time-invariant systems are considered within the oblique projection framework. It is shown that it is inherently not possible to satisfy all the necessary conditions for the local minimizer in the oblique projection framework. The conditions for exact satisfaction of the optimality conditions are also discussed. Further, the equivalence between the tangential interpolation conditions and the gramians-based necessary condition for the local optimum is established. Based on this equivalence, iterative algorithms that nearly satisfy these interpolation-based necessary conditions are proposed. The deviation in satisfaction of the optimality conditions decay as the order of the reduced-model is increased in the proposed algorithms. Moreover, stationary point iteration algorithms that satisfy two out of three necessary conditions for the local minimizer are also proposed. There also, the deviation in satisfaction of the third optimality conditions decay as the order of the reduced-model is increased in the proposed algorithms. The efficacy of the proposed algorithms is validated by considering one illustrative and three high-order models that are considered a benchmark for testing model order reduction algorithms.
\end{abstract}


\section{Introduction} The behaviour of the dynamic systems is expressed as mathematical models comprising of several differential equations. The complexity of modern-day dynamic systems has been increasing at a rapid pace due to scientific innovation and sophistication, resulting in large-scale models, which comprise thousands of differential equations. These models are difficult to simulate and analyze due to the excessive computational demand posed by them. The design procedures, which take these models as input, often end up giving complex solutions that are practically not feasible for implementation. To overcome this challenge, model order reduction (MOR) procedures are used to obtain a reduced-order approximation of the original model. The reduced-order model (ROM) can serve as a surrogate for the original model as its behaviour is close to that of the original model, but it is cheaper to simulate and analyze. The design procedures are also simplified with an admissible numerical accuracy by using the ROMs as surrogates for the original high-order model \citep{benner2005dimension,schilders2008model}. The MOR procedure should preserve important characteristics and properties of the original model. The specific characteristics to be preserved in the ROM lead to various families of MOR algorithms \citep{benner2018model,quarteroni2014reduced}.

The performance of a MOR procedure is judged by computing various system norms of the error transfer function. The $\mathcal{H}_\infty$ norm and the $\mathcal{H}_2$ norm are most widely used as they depict the worst-case scenarios \citep{zhou1996robust}. The $\mathcal{H}_\infty$ norm is a more important norm from a system theory perspective as several performance measures in control and communication systems are expressed in terms of the $\mathcal{H}_\infty$ norm \citep{wolf2014h}. However, the efficient computation of the $\mathcal{H}_\infty$ is a challenge in a large-scale setting. Further, the MOR algorithms that tend to ensure less $\mathcal{H}_\infty$ are generally computationally expensive. In particular, an optimal ROM in the $\mathcal{H}_\infty$ is hard to find, and the algorithms available are computationally demanding \citep{castagnotto2017interpolatory}. The MOR procedure should not be a computational challenge in itself; instead, it should reduce the computational cost of the overall experiment. The $\mathcal{H}_2$ norm, on the other hand, can be computed efficiently due to its relation with the system gramians \citep{zhou1996robust}. There are several low-rank methods to efficiently computed these gramians in a large-scale setting \citep{gugercin2003modified,li2002low,penzl1999cyclic}, which makes the $\mathcal{H}_2$ norm a popular choice for accessing the quality of the ROM. Moreover, there exist several algorithms that compute an optimal ROM in the $\mathcal{H}_2$ norm cheaply in a large-scale setting \citep{gugercin2008h_2,van2008h2,xu2011optimal}. Therefore, we also have used the $\mathcal{H}_2$ norm of the error transfer function as a performance measure for the MOR in this paper.

The $\mathcal{H}_2$-optimal MOR problem is to find a local optimum for the (squared) $\mathcal{H}_2$ norm of the error transfer function. The interpolation-based algorithm that computes a local optimum for single-input single-output (SISO) systems was first proposed in \citep{gugercin2008h_2}, which was later generalized for multi-input multi-output (MIMO) systems in \citep{van2008h2}. A more general algorithm based on Sylvester equations was presented in \citep{xu2011optimal}, which was further improved in \citep{bennersparse}. These algorithms are iterative algorithms with no guarantee of convergence. Generally, these algorithms quickly converge for SISO systems, but the convergence slows down as the number of inputs and outputs increases in MIMO systems. Some trust region-based methods to speed up convergence in these methods are also reported in the literature like the one in \citep{beattie2009trust}.

Some suboptimal methods for MOR in the $\mathcal{H}_2$ norm are also reported in the literature that satisfy a subset of the optimality conditions while preserving some additional properties like stability, cf. \citep{gugercin2008iterative,ibrir2018projection,wolf2013h}. Unlike the algorithm in \citep{gugercin2008iterative}, the pseudo-optimal rational Krylov (PORK) algorithm \citep{wolf2013h} is not an iterative algorithm, and it generates the ROM in a single run. Moreover, unlike the algorithm in \citep{ibrir2018projection}, PORK does not require the solution of any linear matrix inequality (LMI) to satisfy a subset of the optimality conditions. The stability of the ROM in PORK is not only guaranteed, but it has pole-placement property, i.e., the user can select the location of poles on the ROM.

The frequency- and time-limited MOR problems were first introduced in \citep{gawronski1990model}. In these problems, it is aimed to achieve superior accuracy in the frequency and time intervals specified by the user. Since no practical system or simulation is run over infinite frequency and time ranges, the frequency- and time-limited MOR problems are important problems. Thus several algorithms in this category are reported in the literature like \citep{li2014frequency,du2010h,kurschner2018balanced,tahavori2013model}.

Recently, the $\mathcal{H}_2$-optimal MOR in limited frequency and time intervals has received a lot of attention. For the frequency-limited $\mathcal{H}_2$-optimal MOR problem, the gramian-based optimality conditions are derived in \citep{petersson2013nonlinear}, and the interpolation-based optimality conditions are derived in \citep{vuillemin2014frequency}. Similarly, for the time-limited $\mathcal{H}_2$-optimal MOR problem, the gramians-based optimality conditions  are derived in \citep{goyal2019time}, and the interpolation-based optimality conditions are derived in \citep{sinani2019h2}. The first group of algorithms, in the category of $\mathcal{H}_2$-optimal MOR in limited frequency and time intervals, are the nonlinear optimization-based algorithms reported in \citep{petersson2014model,vuillemin2014frequency,sinani2019h2}. These algorithms are applicable to the order of few hundreds, after which, they become computationally infeasible. The second group comprises heuristic generalizations of the standard $\mathcal{H}_2$-optimal MOR algorithms reported in \citep{vuillemin2013h2} and \citep{goyal2019time}. These algorithms are computationally efficient, but they do not satisfy any optimality conditions. The third group comprises suboptimal MOR algorithms reported in \citep{zulfiqar2019adaptive,zulfiqar2020frequency,zulfiqar2020time} that satisfy a subset of the optimality conditions. Unlike the first group, the second and third groups are the oblique projection-based algorithms.

In this paper, we further advance the oblique projection-based framework for the $\mathcal{H}_2$-optimal MOR in limited frequency and time intervals. We show that it is inherently not possible to satisfy the complete set of optimality conditions within the oblique projection framework. Further, we show that the heuristic generalizations of standard $\mathcal{H}_2$-optimal MOR algorithms reported in \citep{vuillemin2013h2} and \citep{goyal2019time} nearly satisfy the optimality conditions. The deviations in satisfaction of the optimality conditions decay with an increase in the order of the ROM. The equivalence between the gramians-based conditions in \citep{goyal2019time} and the interpolation-based conditions in \citep{sinani2019h2} is also established. Further, iterative tangential interpolation algorithms for both the frequency-limited and time-limited MOR problems are proposed that nearly satisfy the optimality conditions upon convergence. Stationary point algorithms for both the frequency-limited and time-limited MOR problems are proposed that achieve a subset of the optimality conditions upon convergence. The efficacy of the proposed algorithms is demonstrated by using illustrative and benchmark numerical examples. The numerical simulation confirms the significance of the algorithms and theoretical results presented in the paper.
\section{Preliminaries}
This section covers the preliminary details of the $\mathcal{H}_2$-optimal MOR in limited frequency and time intervals. The mathematical notations used throughout the text are given in Table \ref{tab0}.
\begin{table}[!h]
\centering
\caption{Mathematical Notations}\label{tab0}
\begin{tabular}{|c|c|}
\hline
Notation & Meaning \\ \hline
$\begin{bmatrix}\cdot\end{bmatrix}^*$  & Hermitian\\
$tr(\cdot)$  & Trace\\
$\mathscr{L}[\cdot]$ & Fr{\'e}chet derivative of the matrix logarithm.\\
$Ran(\cdot)$ & Range\\
$orth(\cdot)$& Orthogonal basis\\
$\underset {i=1,\cdots,r}{span}\{\cdot\}$ & Span of the set of $r$ vectors\\
\hline
\end{tabular}
\end{table}
\subsection{Problem setting}
Let $G(s)$ be an $n^{th}$-order $p\times m$ transfer function of a stable linear time-invariant system, which is related to its state-space realization $(A,B,C)$ as
\begin{align}
G(s)=C(sI-A)^{-1}B\nonumber
\end{align}
where $A\in\mathbb{R}^{n\times n}$, $B\in\mathbb{R}^{n\times m}$, and $C\in\mathbb{R}^{p\times n}$. The state-space equations of this realization are given as
\begin{align}
\dot{x}(t)&=Ax(t)+Bu(t),&y(t)&=Cx(t).\nonumber
\end{align}
The MOR problem is to construct an $r^{th}$-order $p\times m$ transfer function $\hat{G}(s)$ that closely approximates $G(s)$ where $r\ll n$. Let $\hat{G}(s)$ be related to its state-space realization $(\hat{A},\hat{B},\hat{C})$ as
\begin{align}
\hat{G}(s)=\hat{C}(sI-\hat{A})^{-1}\hat{B}\nonumber
\end{align}
where $\hat{A}\in\mathbb{R}^{r\times r}$, $\hat{B}\in\mathbb{R}^{r\times m}$, and $\hat{C}\in\mathbb{R}^{p\times r}$. The state-space equations of this realization are given as
\begin{align}
\dot{x_r}(t)&=\hat{A}x_r(t)+\hat{B}u(t),&y_r(t)&=\hat{C}x_r(t).\nonumber
\end{align}

Let $\hat{V}\in\mathbb{R}^{n\times r}$ and $\hat{W}\in\mathbb{R}^{n\times r}$ be the input and output reduction matrices, respectively, which project $G(s)$ onto an $r$-dimensional subspace where $\Pi=\hat{V}\hat{W}^T$ is the oblique projection, $\hat{W}^T\hat{V}=I$, and the columns of $\hat{V}$ span the reduced subspace along the kernel of $\hat{W}^T$. The $r^{th}$-order ROM obtained via projection is given as
\begin{align}
\hat{A}&=\hat{W}^TA\hat{V},& \hat{B}&=\hat{W}^TB, &\hat{C}&=C\hat{V}.\label{1}
\end{align}
The error transfer function $E(s)=G(s)-\hat{G}(s)$ has the following equivalence with its state-space realization $(A_e,B_e,C_e)$
\begin{align}
E(s)=C_e(sI-A_e)^{-1}B_e\nonumber
\end{align} where
\begin{align}
A_e&=\begin{bmatrix}A&0\\0&\hat{A}\end{bmatrix}, & B_e&=\begin{bmatrix}B\\\hat{B}\end{bmatrix},& C_e&=\begin{bmatrix}C&-\hat{C}\end{bmatrix}.\nonumber
\end{align}

Let $P_{e,\omega}$ and $Q_{e,\omega}$ be the frequency-limited controllability and the frequency-limited observability gramians, respectively, of the realization $(A_e,B_e,C_e)$ within the frequency interval $[-\omega,\omega]$ rad/sec \citep{gawronski1990model}, which solve the following Lyapunov equations
\begin{align}
A_eP_{e,\omega}+P_{e,\omega}A_e^T+B_{e,\omega} B_e^T+B_eB_{e,\omega}^T&=0,\nonumber\\
A_e^TQ_{e,\omega}+Q_{e,\omega}A_e+C_{e,\omega}^T C_e+C_e^TC_{e,\omega}&=0\nonumber
\end{align} where
\begin{align}
B_{e,\omega}&=F_\omega[A_e] B_e,\hspace*{1.5cm}C_{e,\omega}=C_eF_\omega[A_e],\nonumber\\
F_\omega[A_e]&=\frac{j}{2\pi}log\big((j\omega I+A_e)(-j\omega I+A_e)^{-1}\big).\nonumber
\end{align}The energy of the impulse response of $E(s)$ within the frequency interval $[-\omega,\omega]$ rad/sec is quantified by the frequency-limited $\mathcal{H}_2$-norm \citep{petersson2014model}, which is related to $P_{e,\omega}$ and $Q_{e,\omega}$ as
\begin{align}
||E(s)||_{\mathcal{H}_{2,\omega}}&=\sqrt{tr(C_eP_{e,\omega}C_e^T)}=\sqrt{tr(CP_\omega C^T-2C\bar{P}_\omega\hat{C}^T+\hat{C}\hat{P}_\omega\hat{C}^T)}\nonumber\\
&=\sqrt{tr(B_e^TQ_{e,\omega}B_e)}=\sqrt{tr(B^TQ_\omega B-2B^T\bar{Q}_\omega\hat{B}+\hat{B}^T\hat{Q}_\omega\hat{B})}.\nonumber
\end{align}
$P_\omega$, $Q_\omega$, $\bar{P}_\omega$, $\bar{Q}_\omega$, $\hat{P}_\omega$, and $\hat{Q}_\omega$ solve the following linear matrix equations
\begin{align}
AP_\omega+P_\omega A^T+BB_\omega^T+B_\omega B^T&=0,\label{u1}\\
A^TQ_\omega+Q_\omega A+C^T C_\omega+C_\omega^T C&=0,\label{u2}\\
A\bar{P}_\omega+\bar{P}_\omega\hat{A}^T+B\hat{B}_\omega^T+B_\omega \hat{B}^T&=0,\label{d1}\\
A^T\bar{Q}_\omega+\bar{Q}_\omega\hat{A}+C^T\hat{C}_\omega+C_\omega^T \hat{C}&=0,\label{d2}\\
\hat{A}\hat{P}_\omega+\hat{P}_\omega\hat{A}^T+\hat{B}\hat{B}_\omega^T+\hat{B}_\omega \hat{B}^T&=0,\label{d3}\\
\hat{A}^T\hat{Q}_\omega+\hat{Q}_\omega\hat{A}+\hat{C}^T\hat{C}_\omega+\hat{C}_\omega^T \hat{C}&=0\label{d4}
\end{align} where
\begin{align}
B_\omega&=F_\omega[A]B,&\hat{B}_\omega&=F_\omega[\hat{A}] \hat{B},&C_\omega&=CF_\omega[A],&\hat{C}_\omega&=\hat{C}F_\omega[\hat{A}].\nonumber
\end{align}

In the frequency-limited MOR problem, the frequency response of $E(s)$ is sought to be small within the desired frequency interval $\Omega=[-\omega,\omega]$ rad/sec, which is generally quantified by the $\mathcal{H}_{2,\omega}$-norm. Thus the $\mathcal{H}_{2,\omega}$-MOR problem under consideration is to construct $\hat{G}(s)$ such that $||E(s)||_{\mathcal{H}_{2,\omega}}$ is small, i.e.,
\begin{align}
\underset{\substack{\hat{G}(s)\\\textnormal{order}=r}}{\text{min}}||E(s)||_{\mathcal{H}_{2,\omega}}.\nonumber
\end{align} This, in turn, ensures that $||Y(\nu)-Y_r(\nu)||$ is small within the desired frequency interval $\Omega=[-\omega,\omega]$ rad/sec where $Y(\nu)$ and $Y_r(\nu)$ are the Fourier transforms of $y(t)$ and $y_r(t)$, respectively.

Let $P_{e,\tau}$ and $Q_{e,\tau}$ be the time-limited controllability and the time-limited observability gramians, respectively, of the realization $(A_e,B_e,C_e)$ within the time interval $\mathcal{T}=[0,\tau]$ sec \citep{gawronski1990model}, which solve the following Lyapunov equations
 \begin{align}
 A_eP_{e,\tau}+P_{e,\tau}A_e^T+ B_e B_e^T -B_{e,\tau} B_{e,\tau}^T&=0,\nonumber\\
 A_e^TQ_{e,\tau}+Q_{e,\tau}A_e+C_e^TC_e-C_{e,\tau}^TC_{e,\tau}&=0.\nonumber
 \end{align}
 where
 \begin{align}
 B_{e,\tau}&=e^{A_e\tau}B_e&&\textnormal{and} &C_{e,\tau}&=C_ee^{A_e\tau}.\nonumber
 \end{align}
 The energy of the impulse response of $E(s)$ within the time interval $\mathcal{T}=[0,\tau]$ sec is quantified by the time-limited $\mathcal{H}_2$-norm \citep{goyal2019time}, which is related to $P_{e,\tau}$ and $Q_{e,\tau}$ as
\begin{align}
||E(s)||_{\mathcal{H}_{2,\tau}}&=\sqrt{tr(C_eP_{e,\tau}C_e^T)}=\sqrt{tr(CP_\tau C^T-2C\bar{P}_\tau\hat{C}^T+\hat{C}\hat{P}_\tau\hat{C}^T)}\nonumber\\
&=\sqrt{tr(B_e^TQ_{e,\tau}B_e)}=\sqrt{tr(B^TQ_\tau B-2B^T\bar{Q}_\tau\hat{B}+\hat{B}^T\hat{Q}_\tau\hat{B})}.\nonumber
\end{align}
$P_\tau$, $Q_\tau$, $\bar{P}_\tau$, $\bar{Q}_\tau$, $\hat{P}_\tau$, and $\hat{Q}_\tau$ solve the following linear matrix equations
\begin{align}
AP_\tau+P_\tau A^T+BB^T-B_\tau B_{\tau}^T&=0,\label{v1}\\
A^TQ_\tau+Q_\tau A+C^TC-C_{\tau}^T C_{\tau}&=0,\label{v2}\\
A\bar{P}_\tau+\bar{P}_\tau\hat{A}^T+B\hat{B}^T-B_\tau \hat{B}_{\tau}^T&=0,\label{e1}\\
A^T\bar{Q}_\tau+\bar{Q}_\tau\hat{A}+C^T\hat{C}-C_{\tau}^T \hat{C}_{\tau}&=0\label{e2}\\
\hat{A}\hat{P}_\tau+\hat{P}_\tau\hat{A}^T+\hat{B}\hat{B}^T-\hat{B}_\tau \hat{B}_{\tau}^T&=0,\label{e3}\\
\hat{A}^T\hat{Q}_\tau+\hat{Q}_\tau\hat{A}+\hat{C}^T\hat{C}-\hat{C}_{\tau}^T \hat{C}_{\tau}&=0\label{e4}
\end{align} where
\begin{align}
B_\tau&=e^{A\tau}B,&\hat{B}_\tau&=e^{\hat{A}\tau}\hat{B},&C_\tau&=Ce^{A\tau},&\hat{C}_\tau&=\hat{C}e^{\hat{A}\tau}.\nonumber
\end{align}

In the time-limited MOR problem, the time response of $E(s)$ is sought to be small within the desired time interval $\mathcal{T}=[0,\tau]$ sec. This is generally quantified by the $\mathcal{H}_{2,\tau}$-norm. Thus the $\mathcal{H}_{2,\tau}$-MOR problem under consideration is to construct $\hat{G}(s)$ such that $||E(s)||_{\mathcal{H}_{2,\tau}}$ is small, i.e.,
\begin{align}
\underset{\substack{\hat{G}(s)\\\textnormal{order}=r}}{\text{min}}||E(s)||_{\mathcal{H}_{2,\tau}}.\nonumber
\end{align}
This, in turn, ensures that $||y(t)-y_r(t)||$ is small within the desired time interval $\mathcal{T}=[0,\tau]$ sec.
\subsection{Necessary conditions for the local minimizer}
Let $\hat{Q}$ be the observability gramian of the pair $(\hat{A},\hat{C})$, and $\bar{Q}$ solve the following Sylvester equation
 \begin{align}
 A^T\bar{Q}+\bar{Q}\hat{A}+C^T\hat{C}&=0.\nonumber
 \end{align}
Then the local minimizer $(\hat{A},\hat{B},\hat{C})$ for $||E(s)||_{\mathcal{H}_{2,\omega}}^2$ satisfies the following gramian-based conditions \citep{petersson2014model}
\begin{align}
 \bar{Q}^T\bar{P}_\omega-\hat{Q}\hat{P}_\omega+Z_\omega&=0,\label{a1}\\
 \bar{Q}_{\omega}^TB-\hat{Q}_\omega\hat{B}&=0,\label{a2}\\
 C\bar{P}_\omega-\hat{C}\hat{P}_\omega&=0\label{a3}
 \end{align} where
 \begin{align}
 Z_\omega=Re\big(\frac{j}{\pi}\mathscr{L}(-\hat{A}-j\omega I,\hat{C}^T\hat{C}\hat{P}_{\omega}-\hat{C}^TC\bar{P}_{\omega})\big).\nonumber
 \end{align}
Let $G(s)$ and $\hat{G}(s)$ have simple poles and the following pole-residue forms
\begin{align}
G(s)&=\sum_{i=1}^{n}\frac{l_ir_i^T}{s-\lambda_i},&&&\hat{G}(s)&=\sum_{i=1}^{r}\frac{\hat{l}_i\hat{r}_i^T}{s-\hat{\lambda}_i}.\nonumber
\end{align}
Now define $G_\omega(s)$, $T_\omega(s)$, $\hat{G}_\omega(s)$ and $\hat{T}_\omega(s)$ as
\begin{align}
G_\omega(s)&=\sum_{i=1}^{n}\frac{l_ir_i^T}{s-\lambda_i}F_\omega[\lambda_i],&&&T_\omega(s)&=G_\omega(s)+G(s)F_\omega[-s],\nonumber\\
\hat{G}_\omega(s)&=\sum_{i=1}^{r}\frac{\hat{l}_i\hat{r}_i^T}{s-\hat{\lambda}_i}F_\omega[\hat{\lambda}_i],&&&\hat{T}_\omega(s)&=\hat{G}_\omega(s)+\hat{G}(s)F_\omega[-s].\nonumber
\end{align}
Then the local minimizer $\hat{G}(s)$ for $||E(s)||_{\mathcal{H}_{2,\omega}}^2$ satisfies the following bi-tangential Hermite interpolation conditions \citep{vuillemin2014frequency}
\begin{align}
\hat{l}_i^TT^{\prime}_\omega(-\hat{\lambda}_i)\hat{r}_i&=\hat{l}_i^T\hat{T}^{\prime}_\omega(-\hat{\lambda}_i)\hat{r}_i,\label{f1}\\
\hat{l}_i^TT_\omega(-\hat{\lambda}_i)&=\hat{l}_i^T\hat{T}_\omega(-\hat{\lambda}_i),\label{f2}\\
T_\omega(-\hat{\lambda}_i)\hat{r}_i&=\hat{T}_\omega(-\hat{\lambda}_i)\hat{r}_i\label{f3}.
\end{align}

Similarly, the local minimizer $(\hat{A},\hat{B},\hat{C})$ for $||E(s)||_{\mathcal{H}_{2,\tau}}^2$ satisfies the following gramians-based conditions \citep{goyal2019time}
 \begin{align}
 \bar{Q}^T\bar{P}_\tau-\hat{Q}\hat{P}_\tau+Z_\tau&=0,\label{b1}\\
 \bar{Q}_{\tau}^TB-\hat{Q}_\tau\hat{B}&=0,\label{b2}\\
 C\bar{P}_\tau-\hat{C}\hat{P}_\tau&=0\label{b3}
 \end{align} where
  \begin{align}
 Z_\tau=\tau\big(\hat{Q}e^{\hat{A}\tau}\hat{B}\hat{B}^Te^{\hat{A}^T\tau}-\bar{Q}^Te^{A\tau}B\hat{B}^Te^{\hat{A}^T\tau}\big).\nonumber
 \end{align}
Assuming that $G(s)$ and $\hat{G}(s)$ have simple poles, define $G_\tau(s)$, $T_\tau(s)$, $\hat{G}_\tau(s)$, and $\hat{T}_\tau(s)$ as
 \begin{align}
 G_\tau(s)&=-e^{-s\tau}C(sI-A)^{-1}e^{A\tau}B,&&&T_\tau(s)&=G_\tau(s)+G(s),\nonumber\\
 \hat{G}_\tau(s)&=-e^{-s\tau}\hat{C}(sI-\hat{A})^{-1}e^{\hat{A}\tau}\hat{B},&&&\hat{T}_\tau(s)&=\hat{G}_\tau(s)+\hat{G}(s).\nonumber
 \end{align}
 Then the local minimizer $\hat{G}(s)$ for $||E(s)||_{\mathcal{H}_{2,\tau}}^2$ satisfies the following bi-tangential Hermite interpolation conditions \citep{sinani2019h2}
\begin{align}
\hat{l}_i^TT^{\prime}_\tau(-\hat{\lambda}_i)\hat{r}_i&=\hat{l}_i^T\hat{T}^{\prime}_\tau(-\hat{\lambda}_i)\hat{r}_i,\label{g1}\\
\hat{l}_i^TT_\tau(-\hat{\lambda}_i)&=\hat{l}_i^T\hat{T}_\tau(-\hat{\lambda}_i),\label{g2}\\
T_\tau(-\hat{\lambda}_i)\hat{r}_i&=\hat{T}_\tau(-\hat{\lambda}_i)\hat{r}_i.\label{g3}
\end{align}
\section{Existing oblique projection-based MOR techniques}
In this section, some important oblique projection-based algorithms for the $\mathcal{H}_{2,\omega}$- and $\mathcal{H}_{2,\tau}$-MOR problems are briefly reviewed.
 \subsection{Frequency-limited two-sided iteration algorithm (FLTSIA)}\label{sub3.2}
FLTSIA \citep{vuillemin2013h2,du2021frequency} is a heuristic generalization of \citep{xu2011optimal} based on analogy and experimental results. Starting with an initial guess of the ROM, the reduction matrices are updated as $\hat{V}=\bar{P}_\omega$ and $\hat{W}=\bar{Q}_\omega$. To ensure the oblique projection condition $\hat{W}^T\hat{V}=I$, the correction equation $\hat{W}=\hat{W}(\hat{V}^T\hat{W})^{-1}$ is used. The ROM is generated by using these reduction matrices, and the process is repeated until the algorithm converges. It will be shown later in this paper that the ROM constructed by FLTSIA does not satisfy the optimality conditions (\ref{a1})-(\ref{a3}), in general.
 \subsection{Frequency-limited pseudo-optimal rational Krylov algorithm (FLPORK)}
 Let us define $B_\Omega$, $C_\Omega$, $\hat{B}_\Omega$, and $\hat{C}_\Omega$ as
 \begin{align}
 B_\Omega&=\begin{bmatrix}B&B_\omega\end{bmatrix},&&&C_\Omega&=\begin{bmatrix}C\\C_\omega\end{bmatrix},\nonumber\\
 \hat{B}_\Omega&=\begin{bmatrix}\hat{B}&\hat{B}_\omega\end{bmatrix},&&&\hat{C}_\Omega&=\begin{bmatrix}\hat{C}\\\hat{C}_\omega\end{bmatrix}.\nonumber
 \end{align}
 Now define $G_\Omega(s)$, $\hat{G}_\Omega(s)$, $H_\Omega(s)$, and $\hat{H}_\Omega(s)$ as
 \begin{align}
 G_\Omega(s)&=C(sI-A)^{-1}B_\Omega,&&&\hat{G}_\Omega(s)&=\hat{C}(sI-\hat{A})^{-1}\hat{B}_\Omega,\nonumber\\
 H_\Omega(s)&=C_\Omega(sI-A)^{-1}B,&&&\hat{H}_\Omega(s)&=\hat{C}_\Omega(sI-\hat{A})^{-1}\hat{B}.\nonumber
 \end{align}
Assuming that $G(s)$ and $\hat{G}(s)$ have simple poles, it is shown in \citep{zulfiqar2020frequency} that when $\hat{G}(s)$ satisfies the following tangential interpolation conditions
 \begin{align}
 \bar{l}_i^TH_\Omega(-\hat{\lambda}_i)&=\bar{l}_i^T\hat{H}_\Omega(-\hat{\lambda}_i),\label{a21}\\
  G_\Omega(-\hat{\lambda}_i)\bar{r}_i&=\hat{G}_\Omega(-\hat{\lambda}_i)\bar{r}_i\label{a31}
 \end{align} where $\bar{l}_i=\begin{bmatrix}F_\omega[\hat{\lambda}_i]\hat{l}_i\\\hat{l}_i\end{bmatrix}$ and $\bar{r}_i=\begin{bmatrix}F_\omega[\hat{\lambda}_i]\hat{r}_i\\\hat{r}_i\end{bmatrix}$, the optimality conditions (\ref{a2}) and (\ref{a3}) are respectively satisfied. The input reduction matrix $\hat{V}$ in FLPORK \citep{zulfiqar2020frequency} is obtained as
  \begin{align}
	Ran(\hat{V})=\underset {i=1,\cdots,r}{span}\{(\hat{\sigma}_iI-A)^{-1}B_\Omega \bar{b}_i\}\label{e28n}
\end{align} where $\hat{\sigma}_i$ is the interpolation point, $\hat{b}_i$ is the right tangential direction, and $\bar{b}_i=\begin{bmatrix}F_\omega[-\hat{\sigma}_i]\hat{b}_i\\\hat{b}_i\end{bmatrix}$. Now define the oblique projection $\Pi=\hat{V}\hat{W}^T$ wherein $\hat{W}$ is arbitrary. Also, define $B_\bot$, $\hat{S}$, and $\bar{L}$ as
\begin{align}
B_\bot&=(I-\Pi)B_\Omega,&\bar{L}&=(B_\bot^TB_\bot)^{-1}B_\bot(I-\Pi)A\hat{V}&\hat{S}&=\hat{W}^T(A\hat{V}-B_\Omega\bar{L}).\nonumber
\end{align} Further, partition $\bar{L}$ as $\bar{L}=\begin{bmatrix}\hat{L}_\omega\\\hat{L}\end{bmatrix}$ where $\hat{L}_\omega=\hat{L}F_\omega[-\hat{S}]$. The frequency-limited observability gramian $\hat{Q}_{s,\omega}$ of the pair $(-\hat{S},\hat{L})$ solves the following Lyapunov equation
\begin{align}
-\hat{S}^T\hat{Q}_{s,\omega}-\hat{Q}_{s,\omega}\hat{S}+\hat{L}^T\hat{L}_\omega+\hat{L}_\omega^T\hat{L}&=0.\nonumber
\end{align}
Then $\hat{G}(s)$ that satisfies the optimality condition (\ref{a3}) can be obtained as
\begin{align}
\hat{A}&=-\hat{Q}_{s,\omega}^{-1}\hat{S}^T\hat{Q}_{s,\omega},&\hat{B}&=-\hat{Q}_{s,\omega}^{-1}\hat{L}^T,&\hat{C}&=C\hat{V}.\nonumber
\end{align} A dual result also exists that constructs a ROM, which satisfies the optimality condition (\ref{a2}). The output reduction matrix $\hat{W}$ in this case is obtained as
  \begin{align}
	Ran(\hat{W})=\underset {i=1,\cdots,r}{span}\{(\hat{\sigma}_iI-A)^{-T}C_\Omega^T \bar{c}_i^T\}\label{e29n}
\end{align} where $\hat{c}_i$ is the left tangential direction and $\bar{c}_i=\begin{bmatrix}F_\omega[-\hat{\sigma}_i]\hat{c}_i&\hat{c}_i\end{bmatrix}$. Now define the oblique projection $\Pi=\hat{V}\hat{W}^T$ wherein $\hat{V}$ is arbitrary. Also, define $C_\bot$, $\bar{L}$, and $\hat{S}$ as
\begin{align}
C_\bot&=C_\Omega(I-\Pi), &\bar{L}&=\hat{W}^TA(I-\Pi)C_\bot^T(C_\bot C_\bot^T)^{-1},&\hat{S}&=(\hat{W}^TA-\bar{L}C_\Omega)\hat{V}.\nonumber
\end{align}
Further, partition $\bar{L}$ as $\bar{L}=\begin{bmatrix}\hat{L}_\omega&\hat{L}\end{bmatrix}$ where $\hat{L}_\omega=F_\omega[-\hat{S}]\hat{L}$. The frequency-limited controllability gramian $\hat{P}_{s,\omega}$ of the pair $(-\hat{S},\hat{L})$ solves the following Lyapunov equation
\begin{align}
-\hat{S}\hat{P}_{s,\omega}-\hat{P}_{s,\omega}\hat{S}^T+\hat{L}_\omega\hat{L}^T+\hat{L}\hat{L}_\omega^T&=0.\nonumber
\end{align}
Then $\hat{G}(s)$ that satisfies the optimality condition (\ref{a2}) can be obtained as
\begin{align}
\hat{A}&=-\hat{P}_{s,\omega}\hat{S}^T\hat{P}_{s,\omega}^{-1},&\hat{B}&=W^TB,&\hat{C}&=-\hat{L}^T\hat{P}_{s,\omega}^{-1}.\nonumber
\end{align}
\begin{remark}
It was later identified in \citep{zulfiqar2019adaptive} that the interpolation conditions (\ref{a21}) and (\ref{a21}) are the necessary conditions for the local optimum and are equivalent to (\ref{a2}) and (\ref{a3}), respectively, when $G(s)$ and $\hat{G}(s)$ have simple poles.
\end{remark}
\subsection{Time-limited iterative rational Krylov algorithm (TLIRKA)}
TLIRKA \citep{goyal2019time} is a heuristic generalization of the iterative rational Krylov algorithm (IRKA) \citep{gugercin2008h_2} for the time-limited case. Let $\hat{G}(s)$ has simple poles and $\hat{A}=\hat{R}\hat{\Lambda} \hat{R}^{-1}$ be the spectral factorization of $\hat{A}$ where $\hat{\Lambda}=diag(\hat{\lambda}_1,\cdots,\hat{\lambda}_r)$. Starting with a random guess of the ROM, the reduction matrices in TLIRKA are computed as $\hat{V}=\bar{P}_\tau \hat{R}^{-*}$ and $\hat{W}=\bar{Q}_\tau \hat{R}$. To ensure the oblique projection condition $\hat{W}^*\hat{V}=I$, $\hat{V}$ and $\hat{W}$ are updated as $\hat{V}=orth(\hat{V})$, $\hat{W}=orth(\hat{W})$, and $\hat{W}=\hat{W}(\hat{V}^*\hat{W})^{-1}$. The ROM is generated by using these reduction matrices, and the process is repeated until the algorithm converges. In general, The ROM constructed by TLIRKA does not satisfy the optimality conditions (\ref{b1})-(\ref{b3}).
\subsection{Time-limited pseudo-optimal rational Krylov algorithm (TLPORK)}\label{sub3.4}
Let us define $B_{\mathcal{T}}$, $C_{\mathcal{T}}$, $\hat{B}_{\mathcal{T}}$, and $\hat{C}_{\mathcal{T}}$ as
 \begin{align}
 B_{\mathcal{T}}&=\begin{bmatrix}B&-B_\tau\end{bmatrix},&&&C_{\mathcal{T}}&=\begin{bmatrix}C\\-C_\tau\end{bmatrix},\nonumber\\
 \hat{B}_{\mathcal{T}}&=\begin{bmatrix}\hat{B}&-\hat{B}_\tau\end{bmatrix},&&&\hat{C}_{\mathcal{T}}&=\begin{bmatrix}\hat{C}\\-\hat{C}_\tau\end{bmatrix}.\nonumber
 \end{align}
 Now define $G_{\mathcal{T}}(s)$, $\hat{G}_{\mathcal{T}}(s)$, $H_{\mathcal{T}}(s)$, and $\hat{H}_{\mathcal{T}}(s)$ as
 \begin{align}
 G_{\mathcal{T}}(s)&=C(sI-A)^{-1}B_{\mathcal{T}},&&&\hat{G}_{\mathcal{T}}(s)&=\hat{C}(sI-\hat{A})^{-1}\hat{B}_{\mathcal{T}},\nonumber\\
 H_{\mathcal{T}}(s)&=C_{\mathcal{T}}(sI-A)^{-1}B,&&&\hat{H}_{\mathcal{T}}(s)&=\hat{C}_{\mathcal{T}}(sI-\hat{A})^{-1}\hat{B}.\nonumber
 \end{align}
 Assuming that $G(s)$ and $\hat{G}(s)$ have simple poles, it is shown in \citep{zulfiqar2020time} that $\hat{G}(s)$ satisfies the optimality conditions (\ref{b2}) and (\ref{b3}) if the following tangential interpolation conditions are respectively satisfied
 \begin{align}
 \tilde{l}_i^TH_{\mathcal{T}}(-\hat{\lambda}_i)&=\tilde{l}_i^T\hat{H}_{\mathcal{T}}(-\hat{\lambda}_i)\label{b21},\\
  G_{\mathcal{T}}(-\hat{\lambda}_i)\tilde{r}_i&=\hat{G}_{\mathcal{T}}(-\hat{\lambda}_i)\tilde{r}_i\label{b31}
 \end{align} where $\tilde{l}_i=\begin{bmatrix}\hat{l}_i\\e^{\hat{\lambda}_it}\hat{l}_i\end{bmatrix}$ and $\tilde{r}_i=\begin{bmatrix}\hat{r}_i\\e^{\hat{\lambda}_it}\hat{r}_i\end{bmatrix}$. The input reduction matrix $\hat{V}$ in TLPORK \citep{zulfiqar2020time} is obtained as
  \begin{align}
	Ran(\hat{V})=\underset {i=1,\cdots,r}{span}\{(\hat{\sigma}_iI-A)^{-1}B_{\mathcal{T}} \tilde{b}_i\}\label{eq32n}
\end{align} where $\tilde{b}_i=\begin{bmatrix}\hat{b}_i\\e^{-\hat{\sigma}_it}\hat{b}_i\end{bmatrix}$. Now define the oblique projection $\Pi=\hat{V}\hat{W}^T$ wherein $\hat{W}$ is arbitrary. Also, define $B_\bot$, $\hat{S}$, and $\bar{L}$ as
\begin{align}
B_\bot&=(I-\Pi)B_{\mathcal{T}},&\bar{L}&=(B_\bot^TB_\bot)^{-1}B_\bot(I-\Pi)A\hat{V}&\hat{S}&=\hat{W}^T(A\hat{V}-B_{\mathcal{T}}\bar{L}).\nonumber
\end{align} Further, partition $\bar{L}$ as $\bar{L}=\begin{bmatrix}\hat{L}\\\hat{L}_\tau\end{bmatrix}$ where $\hat{L}_\tau=\hat{L}e^{-\hat{S}t}$. The time-limited observability gramian $\hat{Q}_{s,\tau}$ of the pair $(-\hat{S},\hat{L})$ solves the following Lyapunov equation
\begin{align}
-\hat{S}^T\hat{Q}_{s,\tau}-\hat{Q}_{s,\tau}\hat{S}+\hat{L}^T\hat{L}-\hat{L}_\tau^T\hat{L}_\tau&=0.\nonumber
\end{align}
Then $\hat{G}(s)$ that satisfies the optimality condition (\ref{b3}) can be obtained as
\begin{align}
\hat{A}&=-\hat{Q}_{s,\tau}^{-1}\hat{S}^T\hat{Q}_{s,\tau},&\hat{B}&=-\hat{Q}_{s,\tau}^{-1}\hat{L}^T,&\hat{C}&=C\hat{V}.\nonumber
\end{align} A dual result also exists that constructs a ROM, which satisfies the optimality condition (\ref{b2}). The output reduction matrix $\hat{W}$ in this case is obtained as
  \begin{align}
	Ran(\hat{W})=\underset {i=1,\cdots,r}{span}\{(\hat{\sigma}_iI-A)^{-T}C_{\mathcal{T}}^T \tilde{c}_i^T\}\label{eq33n}
\end{align} where $\hat{c}_i$ is the left tangential direction and $\tilde{c}_i=\begin{bmatrix}\hat{c}_i&e^{-\hat{\sigma}_it}\hat{c}_i\end{bmatrix}$. Now define the oblique projection $\Pi=\hat{V}\hat{W}^T$ wherein $\hat{V}$ is arbitrary. Also, define $C_\bot$, $\bar{L}$, and $\hat{S}$ as
\begin{align}
C_\bot&=C_{\mathcal{T}}(I-\Pi), &\bar{L}&=\hat{W}^TA(I-\Pi)C_\bot^T(C_\bot C_\bot^T)^{-1},&\hat{S}&=(\hat{W}^TA-\bar{L}C_{\mathcal{T}})\hat{V}.\nonumber
\end{align}
Further, partition $\bar{L}$ as $\bar{L}=\begin{bmatrix}\hat{L}&\hat{L}_\tau\end{bmatrix}$ where $\hat{L}_\tau=e^{-\hat{S}t}\hat{L}$. The time-limited controllability gramian $\hat{P}_{s,\tau}$ of the pair $(-\hat{S},\hat{L})$ solves the following Lyapunov equation
\begin{align}
-\hat{S}\hat{P}_{s,\tau}-\hat{P}_{s,\tau}\hat{S}^T+\hat{L}\hat{L}^T-\hat{L}_\tau\hat{L}_\tau^T&=0.\nonumber
\end{align}
Then $\hat{G}(s)$ that satisfies the optimality condition () can be obtained as
\begin{align}
\hat{A}&=-\hat{P}_{s,\tau}\hat{S}^T\hat{P}_{s,\tau}^{-1},&\hat{B}&=W^TB,&\hat{C}&=-\hat{L}^T\hat{P}_{s,\tau}^{-1}.\nonumber
\end{align}
\section{$\mathcal{H}_{2,\omega}$- and $\mathcal{H}_{2,\tau}$-optimal MOR}
In this section, the problems of $\mathcal{H}_{2,\omega}$- and $\mathcal{H}_{2,\tau}$-optimal MOR are considered within the oblique projection framework. An inherent difficulty in constructing a local optimum within the oblique projection framework is discussed. The conditions for exact satisfaction of the optimality conditions are also discussed. Further, the equivalence between some gramians-based and interpolation-based optimality conditions is established.
\subsection{Limitation in the oblique projection framework}
Before beginning the discussion on an inherent limitation of the oblique projection framework in satisfying the optimality conditions (\ref{a1})-(\ref{a3}) and (\ref{b1})-(\ref{b3}), let us represent the optimality conditions (\ref{a1}) and (\ref{b1}) a bit differently so that it resembles Wilson's condition \citep{wilson1970optimum} of the standard $\mathcal{H}_2$-optimal MOR case.
\begin{proposition}
Let $F_\omega[\hat{A}]$ has full rank. Then the optimality condition (\ref{a1}) can be expressed as
\begin{align}
\bar{Q}_\omega^T\bar{P}_\omega-\hat{Q}_\omega\hat{P}_\omega+X_\omega=0\nonumber
\end{align} where
\begin{align}
X_\omega&=F_\omega[\hat{A}^T]Z_\omega-\bar{Q}^TF_\omega[A]\bar{P}_\omega+\hat{Q}F_\omega[\hat{A}]\hat{P}_\omega.\nonumber
\end{align}Similarly, the optimality condition (\ref{b1}) can be rewritten as
\begin{align}
\bar{Q}_\tau^T\bar{P}_\tau-\hat{Q}_\tau\hat{P}_\tau+X_\tau&=0\nonumber
\end{align}
where
\begin{align}
X_\tau&=Z_\tau+e^{\hat{A}^T\tau}\bar{Q}^Te^{A\tau}\bar{P}_\tau-e^{\hat{A}^T\tau}\hat{Q}e^{\hat{A}\tau}\hat{P}_\tau.\nonumber
\end{align}
\end{proposition}
\begin{proof}
Note that
\begin{align}
\bar{Q}_\omega^T\bar{P}_\omega&=\big(F_\omega[\hat{A}^T]\bar{Q}^T+\bar{Q}^TF_\omega[A]\big)\bar{P}_\omega&\textnormal{and}&&\hat{Q}_\omega\hat{P}_\omega&=\big(F_\omega[\hat{A}^T]\hat{Q}+\hat{Q}F_\omega[\hat{A}]\big)\hat{P}_\omega.\nonumber
\end{align}
Thus
\begin{align}
F_\omega[\hat{A}^T]\bar{Q}^T\bar{P}_\omega&=\bar{Q}_\omega^T\bar{P}_\omega-\bar{Q}^TF_\omega[A]\bar{P}_\omega,\label{c1}\\
F_\omega[\hat{A}^T]\hat{Q}\hat{P}_\omega&=\hat{Q}_\omega\hat{P}_\omega-\hat{Q}F_\omega[\hat{A}]\hat{P}_\omega.\label{c2}
\end{align}
By pre-multiplying with $F_\omega[\hat{A}^T]$, the equation (\ref{a1}) becomes
\begin{align}
F_\omega[\hat{A}^T]\bar{Q}^T\bar{P}_\omega-F_\omega[\hat{A}^T]\hat{Q}\hat{P}_\omega+F_\omega[\hat{A}^T]Z_\omega=0.\label{c3}
\end{align}
Further, by putting (\ref{c1}) and (\ref{c2}), the equation (\ref{c3}) becomes
\begin{align}
\bar{Q}_\omega^T\bar{P}_\omega-\hat{Q}_\omega\hat{P}_\omega+X_\omega=0.\nonumber
\end{align}
Similarly, note that
\begin{align}
\bar{Q}_\tau^T\bar{P}_\tau&=\big(\bar{Q}^T-e^{\hat{A}^T\tau}\bar{Q}^Te^{A\tau}\big)\bar{P}_\tau&\textnormal{and}&&\hat{Q}_\tau\hat{P}_\tau&=\big(\hat{Q}-e^{\hat{A}^T}\hat{Q}e^{\hat{A}\tau}\big)\hat{P}_\tau.\nonumber
\end{align}
Thus
\begin{align}
\bar{Q}^T\bar{P}_\tau&=\bar{Q}_\tau^T\bar{P}_\tau+e^{\hat{A}^T\tau}\bar{Q}^Te^{A\tau}\bar{P}_\tau,\label{c4}\\
\hat{Q}\hat{P}_\tau&=\hat{Q}_\tau\hat{P}_\tau+e^{\hat{A}^T\tau}\hat{Q}e^{\hat{A}\tau}\hat{P}_\tau.\label{c5}
\end{align}
By putting (\ref{c4}) and (\ref{c5}), the equation (\ref{b1}) becomes
\begin{align}
\bar{Q}_\tau^T\bar{P}_\tau-\hat{Q}_\tau\hat{P}_\tau+X_\tau=0.\nonumber
\end{align}
This completes the proof.
\end{proof}
Let us assume that $\hat{P}_\omega$ and $\hat{Q}_\omega$ are invertible. Then the optimal choices of $\hat{B}$ and $\hat{C}$ according to the optimality conditions (\ref{a2}) and (\ref{a3}), respectively, are given by
\begin{align}
\hat{B}&=\hat{Q}_\omega^{-1}\bar{Q}_\omega^TB&&\textnormal{and}&\hat{C}&=C\bar{P}_\omega\hat{P}_{\omega}^{-1}.\nonumber
\end{align}
Similarly, let us assume that $\hat{P}_\tau$ and $\hat{Q}_\tau$ are invertible. Then the optimal choices of $\hat{B}$ and $\hat{C}$ according to the optimality conditions (\ref{b2}) and (\ref{b3}), respectively, are given by
\begin{align}
\hat{B}&=\hat{Q}_\tau^{-1}\bar{Q}_\tau^TB&&\textnormal{and}&\hat{C}&=C\bar{P}_\tau\hat{P}_{\tau}^{-1}.\nonumber
\end{align}
If ROMs are obtained within the oblique projection framework, the oblique projections for optimal choices $\hat{B}$ and $\hat{C}$ are given by $\Pi=\bar{P}_\omega\hat{P}_{\omega}^{-1}\hat{Q}_\omega^{-1}\bar{Q}_\omega^T$ (with $\hat{V}=\bar{P}_\omega\hat{P}_{\omega}^{-1}$ and $\hat{W}=\bar{Q}_\omega\hat{Q}_\omega^{-1}$) and $\Pi=\bar{P}_\tau\hat{P}_{\tau}^{-1}\hat{Q}_\tau^{-1}\bar{Q}_\tau^T$ (with $\hat{V}=\bar{P}_\tau\hat{P}_{\tau}^{-1}$ and $\hat{W}=\bar{Q}_\tau\hat{Q}_\tau^{-1}$). Further, since $\hat{W}^T\hat{V}=I$, $\bar{Q}_\omega^T\bar{P}_\omega-\hat{Q}_\omega\hat{P}_{\omega}=0$ and $\bar{Q}_\tau^T\bar{P}_\tau-\hat{Q}_\tau\hat{P}_{\tau}=0$. Then deviations in the optimal choices of $\hat{A}$ are given by $X_\omega$ and $X_\tau$. It can readily be verified that $X_\omega=0$ when $F_\omega[A]=\hat{V}F_\omega[\hat{A}]\hat{W}^T$ and $\bar{Q}=\hat{W}\hat{Q}$. Similarly, $X_\tau=0$ when $e^{A\tau}=\hat{V}e^{\hat{A}\tau}\hat{W}^T$ and $\bar{Q}=\hat{W}\hat{Q}$. Note that $\hat{W}\hat{Q}$, $\hat{V}F_\omega[\hat{A}]\hat{W}^T$, and $\hat{V}e^{\hat{A}\tau}\hat{W}^T$ are the oblique projection-based approximations of $\bar{Q}$, $F_\omega[A]$, and $e^{A\tau}$, respectively \citep{benner2016frequency,kurschner2018balanced}. Thus $\bar{Q}\neq\hat{W}\hat{Q}$, $F_\omega[A]\neq\hat{V}F_\omega[\hat{A}]\hat{W}^T$, and $e^{A\tau}\neq\hat{V}e^{\hat{A}\tau}\hat{W}^T$, in general. Resultantly, $X_\omega\neq0$ and $X_\tau\neq0$, in general. Therefore, it is inherently not possible to obtain an optimal choice of $\hat{A}$ within the oblique projection framework when the optimal choices of $\hat{B}$ and $\hat{C}$ are made.
\subsection{Effect of the change of basis}
Note that $\hat{P}_\omega^{-1}$ and $\hat{Q}_\omega^{-1}$ do not change the subspaces but only change the basis of $\bar{P}_\omega$ and $\bar{Q}_\omega$ in $\hat{V}=\bar{P}_\omega\hat{P}_\omega^{-1}$ and $\hat{W}=\bar{Q}_\omega\hat{Q}_\omega^{-1}$, respectively. Similarly, $\hat{P}_\tau^{-1}$ and $\hat{Q}_\tau^{-1}$ only change the basis of $\bar{P}_\tau$ and $\bar{Q}_\tau$ in $\hat{V}=\bar{P}_\tau\hat{P}_\tau^{-1}$ and $\hat{W}=\bar{Q}_\tau\hat{Q}_\tau^{-1}$, respectively \citep{gallivan2004sylvester}. Thus one may think of constructing the ROM using the oblique projections $\Pi=\bar{P}_\omega\bar{Q}_\omega^T$ (as done in FLTSI) and $\Pi=\bar{P}_\tau\bar{Q}_\tau^T$ (as done in TLIRKA). This change of basis is harmless in the standard $\mathcal{H}_2$-optimal MOR; see \citep{bennersparse,xu2011optimal}, for instance. However, in the frequency- and time-limited cases, this incurs deviations in the optimal choices of $\hat{B}$ and $\hat{C}$. The next two theorems show that the deviation caused by this change of basis is zero when $F_\omega[A]=\hat{V}F_\omega[\hat{A}]\hat{W}^T$ and $e^{At}=\hat{V}e^{\hat{A}t}\hat{W}^T$, which is not possible in general.
\begin{theorem}
The ROM $(\hat{A},\hat{B},\hat{C})$ obtained by using the oblique projection $\Pi=\bar{P}_\omega\bar{Q}_\omega^T$ (with $\hat{V}=\bar{P}_\omega$ and $\hat{W}=\bar{Q}_\omega$) satisfies the optimality conditions (\ref{a2}) and (\ref{a3}) provided $B_\omega=\hat{V}\hat{B}_\omega$ and $C_\omega=\hat{C}_\omega\hat{W}^T$.
\end{theorem}
\begin{proof}
By multiplying $\hat{W}^T$ from the left, the equation (\ref{d1}) becomes
\begin{align}
\hat{W}^TA\bar{P}_\omega+\hat{W}^T\bar{P}_\omega\hat{A}^T+\hat{W}^TB\hat{B}_\omega^T+\hat{W}^TB_\omega \hat{B}^T=0.\nonumber
\end{align}
Since $\hat{W}^T\hat{V}=I$ and $\hat{W}^TB_\omega=\hat{B}_\omega$,
\begin{align}
\hat{A}+\hat{A}^T+\hat{B}\hat{B}_\omega^T+\hat{B}_\omega \hat{B}^T=0.\nonumber
\end{align} Due to uniqueness, $\hat{P}_\omega=I$ and thus $C\bar{P}_\omega-C\hat{P}_\omega=0$.

Similarly, by multiplying $\hat{V}^T$ from the left, the equation (\ref{d2}) becomes
\begin{align}
\hat{V}^TA^T\bar{Q}_\omega+\hat{V}^T\bar{Q}_\omega\hat{A}+\hat{V}^TC^T\hat{C}_\omega+\hat{V}^TC_\omega^T \hat{C}=0.\nonumber
\end{align}
Since $\hat{V}^T\hat{W}=I$ and $C_\omega\hat{V}=\hat{C}_\omega$,
\begin{align}
\hat{A}^T+\hat{A}+\hat{C}^T\hat{C}_\omega+\hat{C}_\omega^T \hat{C}=0.\nonumber
\end{align}
Due to uniqueness, $\hat{Q}_\omega=I$ and thus $\bar{Q}_\omega^TB-\hat{Q}_\omega B=0$. This completes the proof.
\end{proof}
$\hat{V}\hat{B}_\omega$ and $\hat{C}_\omega\hat{W}^T$ are projection-based approximations of $B_\omega$ and $C_\omega$, respectively. Thus $B_\omega\neq \hat{V}\hat{B}_\omega$ and $C_\omega\neq \hat{C}_\omega\hat{W}^T$, in general. The Krylov-subspace based methods to obtain approximations of $B_\omega$ and $C_\omega$ as $\hat{V}\hat{B}_\omega$ and $\hat{C}_\omega\hat{W}^T$, respectively, can be found in \citep{benner2016frequency}.
\begin{theorem}
The ROM $(\hat{A},\hat{B},\hat{C})$ obtained by using the oblique projection $\Pi=\bar{P}_\tau\bar{Q}_\tau^T$ (with $\hat{V}=\bar{P}_\tau$ and $\hat{W}=\bar{Q}_\tau$) satisfies the optimality conditions (\ref{b2}) and (\ref{b3}) provided $B_\tau=\hat{V}\hat{B}_\tau$ and $C_\tau=\hat{C}_\tau\hat{W}^T$.
\end{theorem}
\begin{proof}
By multiplying $\hat{W}^T$ from the left, the equation (\ref{e1}) becomes
\begin{align}
\hat{W}^TA\bar{P}_\tau+\hat{W}^T\bar{P}_\tau\hat{A}^T+\hat{W}^TB\hat{B}^T-\hat{W}^TB_\tau \hat{B}_\tau^T=0.\nonumber
\end{align}
Since $\hat{W}^T\hat{V}=I$ and $\hat{W}^TB_\tau=\hat{B}_\tau$,
\begin{align}
\hat{A}+\hat{A}^T+\hat{B}\hat{B}^T-\hat{B}_\tau \hat{B}_\tau^T=0.\nonumber
\end{align} Due to uniqueness, $\hat{P}_\tau=I$ and thus $C\bar{P}_\tau-C\hat{P}_\tau=0$.

Similarly, by multiplying $\hat{V}^T$ from the left, the equation (\ref{e2}) becomes
\begin{align}
\hat{V}^TA^T\bar{Q}_\tau+\hat{V}^T\bar{Q}_\tau\hat{A}+\hat{V}^TC^T\hat{C}-\hat{V}^TC_\tau^T \hat{C}_\tau=0.\nonumber
\end{align}
Since $\hat{V}^T\hat{W}=I$ and $C_\tau\hat{V}=\hat{C}_\tau$,
\begin{align}
\hat{A}^T+\hat{A}+\hat{C}^T\hat{C}-\hat{C}_\tau^T \hat{C}_\tau=0.\nonumber
\end{align}
Due to uniqueness, $\hat{Q}_\tau=I$ and thus $\bar{Q}_\tau^TB-\hat{Q}_\tau B=0$. This completes the proof.
\end{proof}
$\hat{V}\hat{B}_\tau$ and $\hat{C}_\tau\hat{W}^T$ are projection-based approximations of $B_\tau$ and $C_\tau$, respectively. Thus $B_\tau\neq \hat{V}\hat{B}_\tau$ and $C_\tau\neq \hat{C}_\tau\hat{W}^T$, in general. The Krylov-subspace based methods to obtain approximations of $B_\tau$ and $C_\tau$ as $\hat{V}\hat{B}_\tau$ and $\hat{C}_\tau\hat{W}^T$, respectively, can be found in \citep{kurschner2018balanced}.
\subsection{Equivalence between the optimality conditions}
In \citep{zulfiqar2019adaptive}, it is shown that the interpolation conditions (\ref{f2}) and (\ref{f3}) are equivalent to (\ref{a21}) and (\ref{a31}), respectively, when $G(s)$ and $\hat{G}(s)$ have simple poles. In this subsection, we establish equivalence between (some of) the gramians-based optimality conditions in \citep{goyal2019time} and the interpolation-based optimality conditions in \citep{sinani2019h2} for the $\mathcal{H}_{2,\tau}$-optimal MOR problem when $G(s)$ and $\hat{G}(s)$ have simple poles. Further, the interpolation conditions in \citep{sinani2019h2} and \cite{zulfiqar2020time} are shown to be equivalent when $G(s)$ and $\hat{G}(s)$ have simple poles.
 \begin{theorem}
When $G(s)$ and $\hat{G}(s)$ have simple poles, the following statements are true:
\begin{enumerate}
\item The optimality condition (\ref{b3}) is equivalent to the tangential interpolation condition (\ref{b31}).\\
\item The tangential interpolation conditions (\ref{b31}) and (\ref{g3}) are equivalent.\\
\item The optimality condition (\ref{b2}) is equivalent to the tangential interpolation condition (\ref{b21}).\\
\item The tangential interpolation conditions (\ref{b21}) and (\ref{g2}) are equivalent.
\end{enumerate}
\end{theorem}
\begin{proof}
\begin{enumerate}
\item Let us define $\bar{\mathscr{P}}_\tau$ as $\bar{\mathscr{P}}_\tau=\bar{P}_\tau\hat{R}^{-T}$. By noting that $e^{\hat{A}\tau}=\hat{R}e^{\hat{\Lambda}\tau}\hat{R}^{-1}$ and multiplying $\hat{R}^{-T}$ from the right, the equation (\ref{e1}) becomes
\begin{align}
A\bar{\mathscr{P}}_\tau+\bar{\mathscr{P}}_\tau\hat{\Lambda}+B\tilde{B}^T-B_\tau \tilde{B}^Te^{\hat{\Lambda}\tau}=0.\nonumber
\end{align}
Since $\hat{\Lambda}$ (and resultantly $e^{\hat{\Lambda}\tau}$) is a diagonal matrix \citep{goyal2019time}, $\bar{\mathscr{P}}_\tau$ can be computed column-wise as
\begin{align}
\bar{\mathscr{P}}_{\tau,i}=(-\hat{\lambda}_iI-A)^{-1}B\hat{r}_i-(-\hat{\lambda}_iI-A)^{-1}B_\tau e^{\hat{\lambda}_i\tau}\hat{r}_i.\nonumber
\end{align}
Further, let us define $\hat{\mathscr{P}}_\tau$ as $\hat{\mathscr{P}}_\tau=\hat{P}_\tau\hat{R}^{-T}$. By multiplying $\hat{R}^{-T}$ from the right, the equation (\ref{e3}) becomes
\begin{align}
\hat{A}\hat{\mathscr{P}}_\tau+\hat{\mathscr{P}}_\tau\hat{\Lambda}+\hat{B}\tilde{B}^T-\hat{B}_\tau \tilde{B}^Te^{\hat{\Lambda}\tau}=0.\nonumber
\end{align}
Then $\hat{\mathscr{P}}_\tau$ can be computed column-wise as
\begin{align}
\hat{\mathscr{P}}_{\tau,i}=(-\hat{\lambda}_iI-\hat{A})^{-1}\hat{B}\hat{r_i}-(-\hat{\lambda}_iI-\hat{A})^{-1}\hat{B}_\tau e^{\hat{\lambda}_i\tau}\hat{r_i}.\nonumber
\end{align}
Now, the optimality condition (\ref{b3}) can be written as
\begin{align}
C\begin{bmatrix}\bar{\mathscr{P}}_{\tau,1}&\cdots&\bar{\mathscr{P}}_{\tau,r}\end{bmatrix}\hat{R}^T-\hat{C}\begin{bmatrix}\hat{\mathscr{P}}_{\tau,1}&\cdots&\hat{\mathscr{P}}_{\tau,r}\end{bmatrix}\hat{R}^T=0.\label{y1}
\end{align}
After multiplying $\hat{R}^{-T}$ from the right, each column of (\ref{y1}) becomes
\begin{align}
G_{\mathcal{T}}(-\hat{\lambda}_i)\tilde{r}_i-\hat{G}_{\mathcal{T}}(-\hat{\lambda}_i)\tilde{r}_i=0.\nonumber
\end{align}
\item Note that $B_\tau=Re^{\Lambda \tau}R^{-1}B$ and $\hat{B}_\tau=\hat{R}e^{\hat{\Lambda}\tau}\hat{R}^{-1}\hat{B}$. Thus $G_{\mathcal{T}}(s)$ and $\hat{G}_{\mathcal{T}}(s)$ can be represented as
\begin{align}
G_{\mathcal{T}}(s)=\begin{bmatrix}\sum_{i=1}^{n}\frac{l_ir_i^T}{s-\lambda_i}&-\sum_{i=1}^{n}\frac{l_ir_i^T}{s-\lambda_i}e^{\lambda_i \tau}\end{bmatrix}=\begin{bmatrix}G(s)&e^{s\tau}G_\tau(s)\end{bmatrix},\nonumber\\
\hat{G}_{\mathcal{T}}(s)=\begin{bmatrix}\sum_{i=1}^{r}\frac{\hat{l}_i\hat{r}_i^T}{s-\hat{\lambda}_i}&-\sum_{i=1}^{r}\frac{\hat{l}_i\hat{r}_i^T}{s-\hat{\lambda}_i}e^{\hat{\lambda}_i\tau}\end{bmatrix}=\begin{bmatrix}\hat{G}(s)&e^{s\tau}\hat{G}_\tau(s)\end{bmatrix}.\nonumber
\end{align}
Then it can readily be noted that $G_{\mathcal{T}}(-\hat{\lambda}_i)\tilde{r}_i=T_\tau(-\hat{\lambda}_i)\hat{r}_i$ and $\hat{G}_{\mathcal{T}}(-\hat{\lambda}_i)\tilde{r}_i=\hat{T}_\tau(-\hat{\lambda}_i)\hat{r}_i$.
\item Let us define $\bar{\mathscr{Q}}_\tau$ as $\bar{\mathscr{Q}}_\tau=\bar{Q}_\tau\hat{R}$. By multiplying $\hat{R}$ from the right, the equation (\ref{e2}) becomes
\begin{align}
A^T\bar{\mathscr{Q}}_\tau+\bar{\mathscr{Q}}_\tau\hat{\Lambda}+C^T\tilde{C}-C_\tau^T \tilde{C}e^{\hat{\Lambda}\tau}=0.\nonumber
\end{align}
Then $\bar{\mathscr{Q}}_\tau$  can be computed column-wise as
\begin{align}
\bar{\mathscr{Q}}_{\tau,i}=(-\hat{\lambda}_iI-A^T)^{-1}C^T\hat{l}_i-(-\hat{\lambda}_iI-A^T)^{-1}C_\tau^Te^{\hat{\lambda}_i\tau}\hat{l}_i.\nonumber
\end{align}
Now define $\hat{\mathscr{Q}}_\tau$ as $\hat{\mathscr{Q}}_\tau=\hat{Q}_\tau\hat{R}$. By multiplying $\hat{R}$ from the right, the equation (\ref{e4}) becomes
\begin{align}
\hat{A}^T\hat{\mathscr{Q}}_\tau+\hat{\mathscr{Q}}_\tau\hat{\Lambda}+\hat{C}^T\tilde{C}-\hat{C}_\tau^T \tilde{C}e^{\hat{\Lambda}\tau}=0.\nonumber
\end{align}
Then $\hat{\mathscr{Q}}_\tau$  can be computed column-wise as
\begin{align}
\hat{\mathscr{Q}}_{\tau,i}=(-\hat{\lambda}_iI-\hat{A}^T)^{-1}\hat{C}^T\hat{l}_i-(-\hat{\lambda}_iI-\hat{A}^T)^{-1}\hat{C}_\tau^Te^{\hat{\lambda}_i\tau}\hat{l}_i.\nonumber
\end{align}
Further, the optimality condition (\ref{b2}) can be written as
\begin{align}
B^T\begin{bmatrix}\bar{\mathscr{Q}}_{\tau,1}&\cdots&\bar{\mathscr{Q}}_{\tau,r}\end{bmatrix}\hat{R}-\hat{B}^T\begin{bmatrix}\hat{\mathscr{Q}}_{\tau,1}&\cdots&\hat{\mathscr{Q}}_{\tau,r}\end{bmatrix}\hat{R}=0.\label{y2}
\end{align}
After multiplying $\hat{R}^{-1}$ from the right, each column of (\ref{y2}) becomes
\begin{align}
H_{\mathcal{T}}^T(-\hat{\lambda}_i)\tilde{l}_i-\hat{H}_{\mathcal{T}}^T(-\hat{\lambda}_i)\tilde{l}_i=0.\nonumber
\end{align}
\item Note that $C_\tau=CRe^{\Lambda \tau}R^{-1}$ and $\hat{C}_\tau=\hat{C}\hat{R}e^{\hat{\Lambda}\tau}\hat{R}^{-1}$. Thus $H_{\mathcal{T}}(s)$ and $\hat{H}_{\mathcal{T}}(s)$ can be represented as
\begin{align}
H_{\mathcal{T}}(s)&=\begin{bmatrix}\sum_{i=1}^{n}\frac{l_ir_i^T}{s-\lambda_i}\\\sum_{i=1}^{n}\frac{l_ir_i^T}{s-\lambda_i}e^{\lambda_i\tau}\end{bmatrix}=\begin{bmatrix}G(s)\\e^{s\tau}G_\tau(s)\end{bmatrix},\nonumber\\
\hat{H}_{\mathcal{T}}(s)&=\begin{bmatrix}\sum_{i=1}^{r}\frac{\hat{l}_i\hat{r}_i^T}{s-\hat{\lambda}_i}\\\sum_{i=1}^{r}\frac{\hat{l}_i\hat{r}_i^T}{s-\hat{\lambda}_i}e^{\hat{\lambda}_i\tau}\end{bmatrix}=\begin{bmatrix}\hat{G}(s)\\e^{s\tau}\hat{G}_\tau(s)\end{bmatrix}.\nonumber
\end{align}
Then it can readily be noted that $\tilde{l}_i^TH_{\mathcal{T}}(-\hat{\lambda}_i)=\hat{l}_i^TT_\tau(-\hat{\lambda}_i)$ and $\tilde{l}_i^T\hat{H}_{\mathcal{T}}(-\hat{\lambda}_i)=\hat{l}_i^T\hat{T}_\tau(-\hat{\lambda}_i)$.
\end{enumerate}
\end{proof}
\section{$\mathcal{H}_{2,\omega}$- and $\mathcal{H}_{2,\tau}$ near-optimal MOR}
Owing to the difficulty in achieving a local optimum within the oblique projection framework, we focus on achieving near-optimal ROMs for the $\mathcal{H}_{2,\omega}$- and $\mathcal{H}_{2,\tau}$-optimal MOR problems. To that end, we propose two frameworks: interpolation-based framework and stationary point iteration framework. The interpolation-based framework does not satisfy any optimality condition exactly, in general. However, deviations in satisfaction of the optimality conditions decay as the order of ROM grows. The stationary point iteration framework, on the other hand, satisfies two out of three optimality conditions upon convergence, and deviation in satisfaction of the third optimality condition decays as the order of ROM grows.
\subsection{Interpolation framework}
The interpolation theory has been significantly advanced in the last two decades \citep{grimme1997krylov,astolfi2020model}. There are several computational efficient numerical methods to enforce interpolation conditions \citep{beattie2017model}, which has motivated the results of this subsection. We propose iterative tangential interpolation algorithms that target a subset of respective optimality conditions for the $\mathcal{H}_{2,\omega}$- and $\mathcal{H}_{2,\tau}$-optimal MOR problems. The proposed algorithms nearly satisfy their respective optimality conditions upon convergence, and as the order of ROM grows, deviations in satisfaction of the optimality conditions decay further.

We have already seen that the following tangential interpolation conditions are equivalent
\begin{align}
\hat{l}_i^TT_\omega(-\hat{\lambda}_i)&=\hat{l}_i^T\hat{T}_\omega(-\hat{\lambda}_i)&\Longleftrightarrow&&\bar{l}_i^TH_\Omega(-\hat{\lambda}_i)&=\bar{l}_i^T\hat{H}_\Omega(-\hat{\lambda}_i),\nonumber\\
T_\omega(-\hat{\lambda}_i)\hat{r}_i&=\hat{T}_\omega(-\hat{\lambda}_i)\hat{r}_i&\Longleftrightarrow && G_\Omega(-\hat{\lambda}_i)\bar{r}_i&=\hat{G}_\Omega(-\hat{\lambda}_i)\bar{r}_i\nonumber
\end{align} when $G(s)$ and $\hat{G}(s)$ have simple poles. To satisfy these conditions, we need an interpolation framework that preserves the structure of (state-space realizations of) $H_\Omega(s)$ and $G_\Omega(s)$ in $\hat{H}_\Omega(s)$ and $\hat{G}_\Omega(s)$, respectively. The triplet $(-\hat{\lambda}_i,\hat{r}_i,\hat{l}_i)$ is not known apriori, and thus finding such a ROM is a nonconvex problem. We propose an iterative algorithm, wherein starting with an arbitrary guess, the triplet $(-\hat{\lambda}_i,\hat{r}_i,\hat{l}_i)$ is updated in each iteration until convergence. The proposed algorithm is referred to as the ``Frequency-limited iterative tangential interpolation algorithm (FLITIA). The pseudo-code of FLITIA is given in Algorithm \ref{Alg1}. Steps \ref{s1}-\ref{s4} are standard interpolation results (already discussed in Subsection \ref{sub3.2}). Steps \ref{s5}-\ref{s9} are the bi-orthogonal Gram-Schmidt method to ensure the oblique projection condition $\hat{W}^T\hat{V}=I$; cf. \cite{bennersparse}. The ROM $\hat{G}(s)$ is fetched from the ROMs of $G_\Omega(s)$ and $H_\Omega(s)$ in Step \ref{s10}. The interpolation data is updated in Steps \ref{s11}-\ref{s13}. In general, $\tilde{B}_\Omega$ and $\tilde{C}_\Omega$ are not equal to $\hat{B}_\Omega$ and $\hat{C}_\Omega$, respectively, because $\hat{W}^TF_\omega[A]B\neq F_\omega[\hat{A}]\hat{B}$ and $CF_\omega[A]\hat{V}\neq\hat{C}F_\omega[\hat{A}]$. However, as the order of ROM grows, $\hat{V}F_\omega[\hat{A}]\hat{W}^T\approx F_\omega[A]$ because $\hat{V}F_\omega[\hat{A}]\hat{W}^T$ is the oblique projection-based approximation of $F_\omega[A]$. Thus $\tilde{B}_\Omega$ and $\tilde{C}_\Omega$ become nearly equal to $\hat{B}_\Omega$ and $\hat{C}_\Omega$, respectively. Finally, upon convergence, the interpolation conditions (\ref{a21}) and (\ref{a31}) are nearly satisfied.
\begin{algorithm}[!h]
\textbf{Input:} Original system: $(A,B,C)$; Desired frequency interval: $[-\omega,\omega]$ rad/sec; Initial guess: $(\hat{\sigma}_i,\hat{b}_i,\hat{c}_i)$.\\
\textbf{Output:} ROM $(\hat{A},\hat{B},\hat{C})$.
  \begin{algorithmic}[1]
     \STATE Set $B_\Omega=\begin{bmatrix}B&F_\omega[A]B\end{bmatrix}$ and $C_\Omega=\begin{bmatrix}C\\CF_\omega[A]\end{bmatrix}$.\label{s1}
      \STATE \textbf{while} (not converged) \textbf{do}
      \STATE $\bar{b}_i=\begin{bmatrix}F_\omega[-\hat{\sigma}_i]\hat{b}_i\\\hat{b}_i\end{bmatrix}$ and $\bar{c}_i=\begin{bmatrix}F_\omega[-\hat{\sigma}_i]\hat{c}_i&\hat{c}_i\end{bmatrix}$.
      \STATE Compute $\hat{V}$ and $\hat{W}$ from the equations (\ref{e28n}) and (\ref{e29n}), respectively.\label{s4}
      \STATE \textbf{for} $i=1,\ldots,r$ \textbf{do}\label{s5}
      \STATE $\hat{v}=\hat{V}(:,i)$, $\hat{v}=\prod_{k=1}^{i}\big(I-\hat{V}(:,k)\hat{W}(:,k)^T\big)\hat{v}$.
      \STATE $\hat{w}=\hat{W}(:,i)$, $\hat{w}=\prod_{k=1}^{i}\big(I-\hat{W}(:,k)\hat{V}(:,k)^T\big)\hat{w}$.
      \STATE $\hat{v}=\frac{\hat{v}}{||\hat{v}||_2}$, $\hat{w}=\frac{\hat{w}}{||\hat{w}||_2}$, $\hat{v}=\frac{\hat{v}}{\hat{w}^T\hat{v}}$, $\hat{V}(:,i)=\hat{v}$, $\hat{W}(:,i)=\hat{w}$.
      \STATE \textbf{end for}\label{s9}
      \STATE $\hat{A}=\hat{W}^TA\hat{V}$, $\tilde{B}_\Omega=\hat{W}^TB_\Omega=\begin{bmatrix}\hat{B}&\hat{W}^TF_\omega[A]B\end{bmatrix}$, $\tilde{C}_\Omega=C_\Omega\hat{V}=\begin{bmatrix}\hat{C}\\CF_\omega[A]\hat{V}\end{bmatrix}$.\label{s10}
      \STATE Compute spectral factorization of $\hat{A}=\hat{R}\hat{\Lambda}\hat{R}^{-1}$ where $\hat{\Lambda}=diag(\hat{\lambda}_1,\cdots,\hat{\lambda}_r)$.\label{s11}
      \STATE Set $\begin{bmatrix}\hat{r}_1&\cdots&\hat{r}_r\end{bmatrix}=\hat{B}^T\hat{R}^{-*}$ and $\begin{bmatrix}\hat{l}_1&\cdots&\hat{l}_r\end{bmatrix}=\hat{C}\hat{R}$.
      \STATE Update $\hat{\sigma}_i=-\hat{\lambda}_i$, $\hat{b}_i=\hat{r}_i$, and $\hat{c}_i^T=\hat{l}_i$.\label{s13}
      \STATE \textbf{end while}
  \end{algorithmic}
  \caption{FLITIA}\label{Alg1}
\end{algorithm}

Similarly, we have seen that the following tangential interpolation conditions are equivalent
\begin{align}
\hat{l}_i^TT_{\tau}(-\hat{\lambda}_i)&=\hat{l}_i^T\hat{T}_\tau(-\hat{\lambda}_i)&\Longleftrightarrow&&\tilde{l}_i^TH_{\mathcal{T}}(-\hat{\lambda}_i)&=\tilde{l}_i^T\hat{H}_{\mathcal{T}}(-\hat{\lambda}_i),\nonumber\\
T_\tau(-\hat{\lambda}_i)\hat{r}_i&=\hat{T}_\tau(-\hat{\lambda}_i)\hat{r}_i&\Longleftrightarrow && G_{\mathcal{T}}(-\hat{\lambda}_i)\tilde{r}_i&=\hat{G}_{\mathcal{T}}(-\hat{\lambda}_i)\tilde{r}_i.\nonumber
\end{align} when $G(s)$ and $\hat{G}(s)$ have simple poles. To satisfy these conditions, we need an interpolation framework that preserves the structure of (state-space realizations of) $H_{\mathcal{T}}(s)$ and $G_{\mathcal{T}}(s)$ in $\hat{H}_{\mathcal{T}}(s)$ and $\hat{G}_{\mathcal{T}}(s)$, respectively. We propose an iterative algorithm, wherein starting with an arbitrary guess, the triplet $(-\hat{\lambda}_i,\hat{r}_i,\hat{l}_i)$ is updated in each iteration until convergence. The proposed algorithm is referred to as the ``Time-limited iterative tangential interpolation algorithm (TLITIA). The pseudo-code of TLITIA is given in Algorithm \ref{Alg2}. Steps \ref{2as1}-\ref{2as4} are standard interpolation results (already discussed in Subsection \ref{sub3.4}). Steps \ref{2as5}-\ref{2as9} are the bi-orthogonal Gram-Schmidt method to ensure the oblique projection condition $\hat{W}^T\hat{V}=I$. The ROM $\hat{G}(s)$ is fetched from the ROMs of $G_{\mathcal{T}}(s)$ and $H_{\mathcal{T}}(s)$ in Step \ref{2as10}. The interpolation data is updated in Steps \ref{2as11}-\ref{2as13}. In general, $\tilde{B}_{\mathcal{T}}$ and $\tilde{C}_{\mathcal{T}}$ are not equal to $\hat{B}_{\mathcal{T}}$ and $\hat{C}_{\mathcal{T}}$, respectively, because $\hat{W}^Te^{A\tau}B\neq e^{\hat{A}\tau}\hat{B}$ and $Ce^{A\tau}\hat{V}\neq\hat{C}e^{\hat{A}\tau}$. Thus the structures of $B_{\mathcal{T}}$ and $C_{\mathcal{T}}$ are not preserved in $\hat{W}^TB_{\mathcal{T}}$ and $C_{\mathcal{T}}\hat{V}$, respectively. However, as the order of ROM grows, $\hat{V}e^{\hat{A}t}\hat{W}^T\approx e^{At}$ because $\hat{V}e^{\hat{A}t}\hat{W}^T$ is the oblique projection-based approximation of $e^{At}$. Thus $\tilde{B}_{\mathcal{T}}$ and $\tilde{C}_{\mathcal{T}}$ become nearly equal to $\hat{B}_{\mathcal{T}}$ and $\hat{C}_{\mathcal{T}}$, respectively. Finally, upon convergence, the interpolation conditions (\ref{b21}) and (\ref{b31}) are nearly satisfied.
\begin{algorithm}[!h]
\textbf{Input:} Original system: $(A,B,C)$; Desired time interval: $[0,\tau]$ sec; Initial guess: $(\hat{\sigma}_i,\hat{b}_i,\hat{c}_i)$.\\
\textbf{Output:} ROM $(\hat{A},\hat{B},\hat{C})$.
  \begin{algorithmic}[1]
     \STATE Set $B_{\mathcal{T}}=\begin{bmatrix}B&-e^{A\tau}B\end{bmatrix}$ and $C_{\mathcal{T}}=\begin{bmatrix}C\\-Ce^{A\tau}\end{bmatrix}$.\label{2as1}
      \STATE \textbf{while} (not converged) \textbf{do}
      \STATE $\bar{b}_i=\begin{bmatrix}\hat{b}_i\\e^{-\sigma_i\tau}\hat{b}_i\end{bmatrix}$ and $\bar{c}_i=\begin{bmatrix}\hat{c}_i&e^{-\sigma_i\tau}\hat{c}_i\end{bmatrix}$.
      \STATE Compute $\hat{V}$ and $\hat{W}$ from the equations (\ref{eq32n}) and (\ref{eq33n}), respectively.\label{2as4}
      \STATE \textbf{for} $i=1,\ldots,r$ \textbf{do}\label{2as5}
      \STATE $\hat{v}=\hat{V}(:,i)$, $\hat{v}=\prod_{k=1}^{i}\big(I-\hat{V}(:,k)\hat{W}(:,k)^T\big)\hat{v}$.
      \STATE $\hat{w}=\hat{W}(:,i)$, $\hat{w}=\prod_{k=1}^{i}\big(I-\hat{W}(:,k)\hat{V}(:,k)^T\big)\hat{w}$.
      \STATE $\hat{v}=\frac{\hat{v}}{||\hat{v}||_2}$, $\hat{w}=\frac{\hat{w}}{||\hat{w}||_2}$, $\hat{v}=\frac{\hat{v}}{\hat{w}^T\hat{v}}$, $\hat{V}(:,i)=\hat{v}$, $\hat{W}(:,i)=\hat{w}$.
      \STATE \textbf{end for}\label{2as9}
      \STATE $\hat{A}=\hat{W}^TA\hat{V}$, $\tilde{B}_{\mathcal{T}}=\hat{W}^TB_{\mathcal{T}}=\begin{bmatrix}\hat{B}&-\hat{W}^Te^{A\tau}B\end{bmatrix}$, $\tilde{C}_{\mathcal{T}}=C_{\mathcal{T}}\hat{V}=\begin{bmatrix}\hat{C}\\-Ce^{A\tau}\hat{V}\end{bmatrix}$.\label{2as10}
      \STATE Compute spectral factorization of $\hat{A}=\hat{R}\hat{\Lambda}\hat{R}^{-1}$ where $\hat{\Lambda}=diag(\hat{\lambda}_1,\cdots,\hat{\lambda}_r)$.\label{2as11}
      \STATE Set $\begin{bmatrix}\hat{r}_1&\cdots&\hat{r}_r\end{bmatrix}=\hat{B}^T\hat{R}^{-*}$ and $\begin{bmatrix}\hat{l}_1&\cdots&\hat{l}_r\end{bmatrix}=\hat{C}\hat{R}$.
      \STATE Update $\hat{\sigma}_i=-\hat{\lambda}_i$, $\hat{b}_i=\hat{r}_i$, and $\hat{c}_i^T=\hat{l}_i$.\label{2as13}
      \STATE \textbf{end while}
  \end{algorithmic}
  \caption{TLITIA}\label{Alg2}
\end{algorithm}
\subsection{Stationary point iteration framework}
Owing to the inherent difficulty in making an optimal choice of $\hat{A}$ within the oblique projection framework, we now focus on constructing a ROM $(\hat{A},\hat{B},\hat{C})$ wherein $\hat{B}$ and $C$ are the optimal choices. Let the ROM be obtained by using the oblique projection $\Pi=\bar{P}_\omega\hat{P}_{\omega}^{-1}\hat{Q}_\omega^{-1}\bar{Q}_\omega^T$ (with $\hat{V}=\bar{P}_\omega\hat{P}_{\omega}^{-1}$ and $\hat{W}=\bar{Q}_\omega\hat{Q}_\omega^{-1}$). Since $\bar{P}_\omega$, $\hat{P}_\omega$, $\bar{Q}_\omega$, and $\hat{Q}_\omega$ depend on the ROM $(\hat{A},\hat{B},\hat{C})$, finding such a projection is a nonconvex problem. Note that (\ref{1}) and the equations (\ref{d1})-(\ref{d4}) can be seen as the following coupled system of equations
\begin{align}
(\hat{A},\hat{B},\hat{C})&=f_\omega(\bar{P}_\omega,\hat{P}_\omega,\bar{Q}_\omega,\hat{Q}_\omega),&&&(\bar{P}_\omega,\hat{P}_\omega,\bar{Q}_\omega,\hat{Q}_\omega)&=g_\omega(\hat{A},\hat{B},\hat{C}).\nonumber
\end{align}
The stationary points of $(\hat{A},\hat{B},\hat{C})=f_\omega\big(g_\omega(\hat{A},\hat{B},\hat{C})\big)$ satisfy the optimality conditions (\ref{a2}) and (\ref{a3}). The pseudo-code of the stationary point iteration algorithm to compute stationary points of $(\hat{A},\hat{B},\hat{C})=f_\omega\big(g_\omega(\hat{A},\hat{B},\hat{C})\big)$ is given in Algorithm \ref{Alg3}, which is referred to as the ``Frequency-limited $\mathcal{H}_2$-suboptimal MOR (FLHMOR)". The oblique projection condition $\hat{W}^T\hat{V}=I$ is ensured by using the biorthogonal Gram-Schmidt method, i.e., steps \ref{sta6}-\ref{sta11}.
\begin{algorithm}[!h]
\textbf{Input:} Original system: $(A,B,C)$; Desired frequency interval: $[-\omega,\omega]$ rad/sec; Initial guess: $(\hat{A},\hat{B},\hat{C})$.\\
\textbf{Output:} ROM $(\hat{A},\hat{B},\hat{C})$.
  \begin{algorithmic}[1]
      \STATE \textbf{while} (not converged) \textbf{do}
      \STATE Solve the equations (\ref{d1}) and (\ref{d2}) to compute $\bar{P}_{\omega}$ and $\bar{Q}_\omega$, respectively.
      \STATE Solve the equations (\ref{d3}) and (\ref{d4}) to compute $\hat{P}_{\omega}$ and $\hat{Q}_\omega$, respectively.
      \STATE Set $\hat{V}=\bar{P}_{\omega}\hat{P}_{\omega}^{-1}$ and $\hat{W}=\bar{Q}_{\omega}\hat{Q}_{\omega}^{-1}$.
      \STATE \textbf{for} $i=1,\ldots,r$ \textbf{do}\label{sta6}
      \STATE $\hat{v}=\hat{V}(:,i)$, $\hat{v}=\prod_{k=1}^{i}\big(I-\hat{V}(:,k)\hat{W}(:,k)^T\big)\hat{v}$.
      \STATE $\hat{w}=\hat{W}(:,i)$, $\hat{w}=\prod_{k=1}^{i}\big(I-\hat{W}(:,k)\hat{V}(:,k)^T\big)\hat{w}$.
      \STATE $\hat{v}=\frac{\hat{v}}{||\hat{v}||_2}$, $\hat{w}=\frac{\hat{w}}{||\hat{w}||_2}$, $\hat{v}=\frac{\hat{v}}{\hat{w}^T\hat{v}}$, $\hat{V}(:,i)=\hat{v}$, $\hat{W}(:,i)=\hat{w}$.
      \STATE \textbf{end for}\label{sta11}
      \STATE $\hat{A}=\hat{W}^TA\hat{V}$, $\hat{B}=\hat{W}^TB$, $\hat{C}=C\hat{V}$.
      \STATE \textbf{end while}
  \end{algorithmic}
  \caption{FLHMOR}\label{Alg3}
\end{algorithm}

Similarly, if the ROM is obtained by using the oblique projections $\Pi=\bar{P}_\tau\hat{P}_{\tau}^{-1}\hat{Q}_\tau^{-1}\bar{Q}_\tau^T$ (with $\hat{V}=\bar{P}_\tau\hat{P}_{\tau}^{-1}$ and $\hat{W}=\bar{Q}_\tau\hat{Q}_\tau^{-1}$), (\ref{1}) and the equations (\ref{e1})-(\ref{e4}) can be seen as the following coupled system of equations
\begin{align}
(\hat{A},\hat{B},\hat{C})&=f_\tau(\bar{P}_\tau,\hat{P}_\tau,\bar{Q}_\tau,\hat{Q}_\tau),&&&(\bar{P}_\tau,\hat{P}_\tau,\bar{Q}_\tau,\hat{Q}_\tau)&=g_\tau(\hat{A},\hat{B},\hat{C}).\nonumber
\end{align}
The stationary points of $(\hat{A},\hat{B},\hat{C})=f_\tau\big(g_\tau(\hat{A},\hat{B},\hat{C})\big)$ satisfy the optimality conditions (\ref{b2}) and (\ref{b3}). The pseudo-code of the stationary point iteration algorithm to compute stationary points of $(\hat{A},\hat{B},\hat{C})=f_\tau\big(g_\tau(\hat{A},\hat{B},\hat{C})\big)$ is given in Algorithm \ref{Alg4}, which is referred to as the ``Time-limited $\mathcal{H}_2$-suboptimal MOR (TLHMOR)". The oblique projection condition $\hat{W}^T\hat{V}=I$ is again ensured by using the biorthogonal Gram-Schmidt method.
\begin{algorithm}[!h]
\textbf{Input:} Original system: $(A,B,C)$; Desired time interval: $[0,\tau]$ sec; Initial guess: $(\hat{A},\hat{B},\hat{C})$.\\
\textbf{Output:} ROM $(\hat{A},\hat{B},\hat{C})$.
  \begin{algorithmic}[1]
      \STATE \textbf{while} (not converged) \textbf{do}
      \STATE Solve the equations (\ref{e1}) and (\ref{e2}) to compute $\bar{P}_{\tau}$ and $\bar{Q}_\tau$, respectively.
      \STATE Solve the equations (\ref{e3}) and (\ref{e4}) to compute $\hat{P}_{\tau}$ and $\hat{Q}_\tau$, respectively.
      \STATE Set $\hat{V}=\bar{P}_{\tau}\hat{P}_{\tau}^{-1}$ and $\hat{W}=\bar{Q}_{\tau}\hat{Q}_{\tau}^{-1}$.
      \STATE \textbf{for} $i=1,\ldots,r$ \textbf{do}
      \STATE $\hat{v}=\hat{V}(:,i)$, $\hat{v}=\prod_{k=1}^{i}\big(I-\hat{V}(:,k)\hat{W}(:,k)^T\big)\hat{v}$.
      \STATE $\hat{w}=\hat{W}(:,i)$, $\hat{w}=\prod_{k=1}^{i}\big(I-\hat{W}(:,k)\hat{V}(:,k)^T\big)\hat{w}$.
      \STATE $\hat{v}=\frac{\hat{v}}{||\hat{v}||_2}$, $\hat{w}=\frac{\hat{w}}{||\hat{w}||_2}$, $\hat{v}=\frac{\hat{v}}{\hat{w}^T\hat{v}}$, $\hat{V}(:,i)=\hat{v}$, $\hat{W}(:,i)=\hat{w}$.
      \STATE \textbf{end for}
      \STATE $\hat{A}=\hat{W}^TA\hat{V}$, $\hat{B}=\hat{W}^TB$, $\hat{C}=C\hat{V}$.
      \STATE \textbf{end while}
  \end{algorithmic}
  \caption{TLHMOR}\label{Alg4}
\end{algorithm}
\section{Computational aspects}
In this section, we briefly discuss some important computational aspects to be considered for efficient implementation of the proposed algorithms.

\textit{Generic frequency and time intervals:} Throughout the paper, the desired frequency and time intervals are assumed for simplicity to be $\Omega=[-\omega,\omega]$ rad/sec and $\mathcal{T}=[0,\tau]$ sec, respectively. With some modifications, the proposed algorithms can be generalized for any generic frequency and time interval, i.e., $\Omega=[-\omega_2,-\omega_1]\cup[\omega_1,\omega_2]$ rad/sec and $\mathcal{T}=[\tau_1,\tau_2]$ sec, respectively. For the generic case, $F_\omega[A]$ and $F_\omega[-\hat{\sigma}_i]$ in FLITIA and FLHMOR are computed as the following
\begin{align}
F_\omega[A]&=F_{\omega_2}[A]-F_{\omega_1}[A]&&\textnormal{and}&F_\omega[A]&=F_{\omega_2}[-\hat{\sigma}_i]-F_{\omega_1}[-\hat{\sigma}_i].\nonumber
\end{align}
For the generic case, $B_{\mathcal{T}}$, $C_{\mathcal{T}}$, $\bar{b}_i$, and $\bar{c}_i$ in TLITIA are computed as
\begin{align}
B_{\mathcal{T}}&=\begin{bmatrix}e^{At_1}B&-e^{At_2}B\end{bmatrix},&&& C_{\mathcal{T}}&=\begin{bmatrix}Ce^{At_1}\\-Ce^{At_2}\end{bmatrix},\nonumber\\
\bar{b}_i&=\begin{bmatrix}e^{-\sigma_it_1}\hat{b}_i\\e^{-\sigma_it_2}\hat{b}_i\end{bmatrix},&&&\bar{c}_i&=\begin{bmatrix}e^{-\sigma_it_1}\hat{c}_i&e^{-\sigma_it_2}\hat{c}_i\end{bmatrix}.\nonumber
\end{align}
On similar lines, for the generic case, $\bar{P}_\tau$, $\hat{P}_\tau$, $\bar{Q}_\tau$, and $\hat{Q}_\tau$ in TLHMOR are computed as the following
\begin{align}
A\bar{P}_\tau+\bar{P}_\tau\hat{A}^T+e^{At_1}B\hat{B}^Te^{\hat{A}^Tt_1}-e^{At_2}B\hat{B}^Te^{\hat{A}^Tt_2}&=0,\nonumber\\
\hat{A}\hat{P}_\tau+\hat{P}_\tau\hat{A}^T+e^{\hat{A}t_1}\hat{B}\hat{B}^Te^{\hat{A}^Tt_1}-e^{\hat{A}t_2}\hat{B}\hat{B}^Te^{\hat{A}^Tt_2}&=0,\nonumber\\
A^T\bar{Q}_\tau+\bar{Q}_\tau\hat{A}+e^{A^Tt_1}C^T\hat{C}e^{\hat{A}t_1}-e^{A^Tt_2}C^T\hat{C}e^{\hat{A}t_2}&=0,\nonumber\\
\hat{A}^T\hat{Q}_\tau+\hat{Q}_\tau\hat{A}+e^{\hat{A}^Tt_1}\hat{C}^T\hat{C}e^{\hat{A}t_1}-e^{\hat{A}^Tt_2}\hat{C}^T\hat{C}e^{\hat{A}t_2}&=0.\nonumber
\end{align}

\textit{Computation of $F_\omega[A]$ and $e^{A\tau}$:} As the order of the original model becomes high, the computation of $F_\omega[A]$ and $e^{A\tau}$ becomes expensive. In this scenario, $F_\omega[A]$ and $e^{A\tau}$ can be approximated as $\hat{V}F_\omega[\hat{A}]\hat{W}^T$ and $\hat{V}e^{\hat{A}\tau}\hat{W}^T$, respectively. The Krylov subspace-based methods to compute $\hat{V}$ and $\hat{W}$ for this purpose can be found in \citep{benner2016frequency,kurschner2018balanced}.

\textit{Computation of $\bar{P}_\omega$, $\bar{Q}_\omega$, $\bar{P}_\tau$, and $\bar{Q}_\tau$:} In general, $p\ll n$, $m\ll n$, and the matrices $A$, $B$, and $C$ are sparse in a large-scale setting. This makes the equations (\ref{d1}), (\ref{d2}), (\ref{e1}), and (\ref{e2}) a special type of Sylvester equation known as ``\textit{sparse-dense}" Sylvester equation in the literature, cf. \citep{bennersparse}, wherein the large-scale matrices are small, and the small-scale matrices are dense. An efficient solver for this class of Sylvester equation that can be used to compute $\bar{P}_\omega$, $\bar{Q}_\omega$, $\bar{P}_\tau$, and $\bar{Q}_\tau$ efficiently in a large-scale setting is proposed in \citep{bennersparse}.

\textit{Approximations of $P_\omega$, $Q_\omega$, $P_\tau$, $Q_\tau$:} If either the condition (\ref{a2}) or (\ref{a3}) is met, the following holds
\begin{align}
||E(s)||_{\mathcal{H}_{2,\omega}}^2=tr\big(C(\bar{P}_\omega-\hat{V}\hat{P}_\omega\hat{V}^T)C^T\big)=tr\big(B^T(\bar{Q}_\omega-\hat{W}\hat{Q}_\omega\hat{W}^T)B\big).\nonumber
\end{align}
Thus, it can readily be concluded that FLITIA and FLHMOR provide approximations of $P_\omega$ and $Q_\omega$ as $\tilde{P}_\omega=\hat{V}\hat{P}_\omega\hat{V}^T$ and $\tilde{Q}_\omega=\hat{W}\hat{Q}_\omega\hat{W}^T$, respectively. $\tilde{P}_\omega$ and $\tilde{Q}_\omega$ can be used to perform approximate frequency-limited balanced truncation (FLBT) \citep{gawronski1990model} without solving the large-scale Lyapunov equations (\ref{u1}) and (\ref{u2}). Therefore, FLITIA and FLHMOR can be used to reduce the computational cost of FLBT.

Similarly, if either the condition (\ref{b2}) or (\ref{b3}) is met, the following holds
\begin{align}
||E(s)||_{\mathcal{H}_{2,\tau}}^2=tr\big(C(\bar{P}_\tau-\hat{V}\hat{P}_\tau\hat{V}^T)C^T\big)=tr\big(B^T(\bar{Q}_\tau-\hat{W}\hat{Q}_\tau\hat{W}^T)B\big).\nonumber
\end{align}
Thus, it can readily be concluded that TLITIA and TLHMOR provide approximations of $P_\tau$ and $Q_\tau$ as $\tilde{P}_\tau=\hat{V}\hat{P}_\tau\hat{V}^T$ and $\tilde{Q}_\tau=\hat{W}\hat{Q}_\tau\hat{W}^T$, respectively. $\tilde{P}_\tau$ and $\tilde{Q}_\tau$ can be used to perform approximate time-limited balanced truncation (TLBT) \citep{gawronski1990model} without solving the large-scale Lyapunov equations (\ref{v1}) and (\ref{v2}). Therefore, TLITIA and TLHMOR can be used to reduce the computational cost of TLBT.
\section{Numerical results}
In this section, the proposed algorithms are tested on four test models. The first model is a small-order illustrative system, which is considered to aid convenient validation of the theoretical results presented in the paper. The other three models are high-order systems taken from the collection of benchmark systems for MOR presented in \citep{chahlaoui2005benchmark}.
\subsection{Experimental setup and hardware}
The selection of desired frequency and time intervals is made arbitrarily for demonstration purposes. The initial guess of the ROM in FLHMOR and TLHMOR is also made arbitrarily. The mirror images of the poles and residues of the respective initial guess are used as interpolation points and tangential directions, respectively, in FLITIA and TLITIA. The ROMs are constructed by using FLITIA, TLITIA, FLHMOR, and TLHMOR. The accuracy of the ROMs is compared with that of FLBT and TLBT, as these are considered gold standards of their respective problems. Although FLTSIA and TLIRKA are heuristic in nature and based on experimental results, they are very similar in essence to FLITIA and TLITIA, respectively. Their performance is also quite similar and often indistinguishable. Therefore, we have not shown the results of FLTSIA and TLIRKA in our simulations for clarity. The Lyapunov and Sylvester equations are solved using MATLAB's \textit{``lyap"} command. The experiments are performed using MATLAB 2016 on a laptop with a 2GHz-i7 Intel processor and 16GB of memory.
\subsection{Illustrative example}
Consider a sixth-order model with the following state-space realization
\begin{align}
A &=\begin{bsmallmatrix}-1.7682&0.4541&0.5078&-0.1395&-0.1218&0.7166\\
   -0.2639&-2.4685&-0.6461&1.3914&0.1109&0.3260\\
    0.1337&-1.1676&-2.2543&0.0382&0.0584&0.4749\\
    0.3569&0.2430&-0.9208&-2.3415&0.2474&-1.1362\\
    0.0093&0.4753&0.1530&-0.2144&-2.1098&0.1888\\
    0.8927&-0.8289&-0.2549&-0.6194&0.4891&-2.0675\end{bsmallmatrix},&B &=\begin{bsmallmatrix}     0     &    0\\
         0    0.5963\\
   -0.1556  & -0.1135\\
    0.1291  &  0.8070\\
         0  & -0.0898\\
   -0.0301  & -0.0063\end{bsmallmatrix},\nonumber\\
   C &=\begin{bsmallmatrix}-0.0919&-0.9212&-0.9270&-0.9612&1.7848&-0.2002\end{bsmallmatrix}.\nonumber
\end{align} The initial guess of the ROM used in FLHMOR and TLHMOR is given by
\begin{align}
\hat{A}^{(0)}&=\begin{bmatrix}-0.5763 & -0.5492\\-0.6734&-2.6844\end{bmatrix}, &\hat{B}^{(0)}&=\begin{bmatrix}-0.2749&-1.0346\\0.1071& 1.6785\end{bmatrix},\nonumber\\
\hat{C}^{(0)}&=\begin{bmatrix}-0.0260&-0.8520\end{bmatrix}.\nonumber
\end{align}
The desired frequency interval in FLHMOR and FLITIA is set to $[-0.5,0.5]$ rad/sec. The ROM constructed by FLHMOR is given by
\begin{align}
\hat{A} &=\begin{bmatrix}-0.5772&-0.7972\\-0.4698&-2.6876\end{bmatrix},&
\hat{B} &=\begin{bmatrix}0.1461&0.5396\\-0.0254&-0.5516\end{bmatrix},\nonumber\\
\hat{C} &=\begin{bmatrix}0.1027&2.6474\end{bmatrix}.\nonumber
\end{align} It can be verified that
\begin{align}
C\bar{P}_\omega&=\hat{C}\hat{P}_\omega=\begin{bmatrix}-0.2169&0.0679\end{bmatrix},&\bar{Q}_\omega^TB&=\hat{Q}_\omega\hat{B}=\begin{bmatrix}0.0143&0.1051\\-0.0221&-0.1778\end{bmatrix}.\nonumber
\end{align}However, $||F_\omega[A]-\hat{V}F_\omega[\hat{A}]\hat{W}^T||_2=0.1502$, which reveals that the optimality condition (\ref{a1}) is not exactly satisfied. The ROM constructed by FLITIA is given by
\begin{align}
\hat{A} &=\begin{bmatrix}-0.6921  & -1.5934\\
   -0.3789 &  -2.5727\end{bmatrix},&
\hat{B} &=\begin{bmatrix}0.0931  &  0.7682\\
   -0.0216  &  0.3427\end{bmatrix},\nonumber\\
\hat{C} &=\begin{bmatrix}-0.9933  & -1.8726\end{bmatrix}.\nonumber
\end{align} Interestingly, FLITIA also satisfies the optimality conditions (\ref{a2}) and (\ref{a3}) exactly in this case, and it can be verified that
\begin{align}
C\bar{P}_\omega&=\hat{C}\hat{P}_\omega=\begin{bmatrix}-0.1662& 0.0041\end{bmatrix},&\bar{Q}_\omega^TB&=\hat{Q}_\omega\hat{B}=\begin{bmatrix}0.0275& 0.2093\\
   -0.0037  & -0.0175\end{bmatrix}.\nonumber
\end{align}However, $||F_\omega[A]-\hat{V}F_\omega[\hat{A}]\hat{W}^T||_2=0.1502$, which reveals that the optimality condition (\ref{a1}) is not exactly satisfied. A closer look at the ROMs constructed by FLHMOR and FLITIA reveals that these are actually different realizations of the same transfer function.

The desired time interval in TLHMOR and TLITIA is set to $[0,0.1]$ sec.
The ROM constructed by TLHMOR is given by
\begin{align}
\hat{A}&=\begin{bmatrix}-2.5168&0.4390\\1.2046&-2.5553\end{bmatrix},&
\hat{B}&=\begin{bmatrix}0.1618&0.2539\\0.0342&0.6215\end{bmatrix},\nonumber\\
\hat{C}&=\begin{bmatrix}0.68987-2.5008\end{bmatrix}.\nonumber
\end{align}
It can be verified that
\begin{align}
C\bar{P}_\tau&=\hat{C}\hat{P}_\tau=\begin{bmatrix}-0.0294&-0.0700\end{bmatrix},&\bar{Q}_\tau^TB&=\hat{Q}_\tau\hat{B}=\begin{bmatrix}0.0003&-0.0608\\
-0.0009&0.2740\end{bmatrix}.\nonumber
\end{align} However, $||e^{At_2}-\hat{V}e^{\hat{A}t_2}\hat{W}^T||_2=1.4127$, which reveals that the optimality condition (\ref{b1}) is not exactly satisfied. The ROM constructed by TLITIA is given by
\begin{align}
\hat{A} &=\begin{bmatrix}-2.4649 &  -1.2215\\
   -0.4291 &  -2.6072\end{bmatrix},&
\hat{B} &=\begin{bmatrix}0.0020 &  -0.5830\\
    0.1609  &  0.1621\end{bmatrix},\nonumber\\
\hat{C} &=\begin{bmatrix}2.4022  &  0.1330
\end{bmatrix}.\nonumber
\end{align} Interestingly, TLITIA also satisfies the optimality conditions (\ref{b2}) and (\ref{b3}) exactly in this case, and it can be verified that
\begin{align}
C\bar{P}_\tau&=\hat{C}\hat{P}_\tau=\begin{bmatrix}0.0655&-0.0190\end{bmatrix},&\bar{Q}_\tau^TB&=\hat{Q}_\tau\hat{B}=\begin{bmatrix}0.0008&-0.2655\\
    0.0001 &   0.0003\end{bmatrix}.\nonumber
\end{align}However, $||e^{At_2}-\hat{V}e^{\hat{A}t_2}\hat{W}^T||_2=1.4127$, which reveals that the optimality condition (\ref{b1}) is not exactly satisfied. A closer look at the ROMs constructed by TLHMOR and TLITIA reveals that these are actually different realizations of the same transfer function.
\subsection{Clamped Beam} Let us consider the $348^{th}$ order clamped beam model as the test system from the benchmark collection of test systems for MOR, cf. \citep{chahlaoui2005benchmark}. For the frequency-limited case, the desired frequency interval is chosen as $[-6,-4]\cup[4,6]$ rad/sec. The clamped beam model is reduced by using FLBT, FLITIA, and FLHMOR, and their accuracy is compared in Table \ref{tab2}. It can be noticed that FLITIA and FLHMOR offer less $||E(s)||_{\mathcal{H}_{2,\omega}}$ than FLBT.
\begin{table}[!h]
\caption{$||E(s)||_{\mathcal{H}_{2,\omega}}$}\label{tab2}
\centering
\begin{tabular}{|c|c|c|c|}
\hline
Order & FLBT & FLITIA & FLHMOR \\ \hline
10    &  $0.0118$    &  $0.0099$      &  $0.0099$      \\ \hline
11    & $0.0203$     &  $0.0099$      &   $0.0099$     \\ \hline
12    &  $4.2345\times 10^{-4}$    & $4.1256\times 10^{-4}$       &   $4.1952\times 10^{-4}$     \\ \hline
13    &  $2.4317\times 10^{-4}$    &  $2.2695\times 10^{-4}$      &   $2.2364\times 10^{-4}$     \\ \hline
14    &  $2.4189\times 10^{-4}$    &  $2.0278\times 10^{-4}$      &  $2.0642\times 10^{-4}$      \\ \hline
15    &  $2.4109\times 10^{-4}$    &   $1.9690\times 10^{-4}$     & $2.0642\times 10^{-4}$       \\ \hline
\end{tabular}
\end{table}The frequency-domain error for the ROMs of order $14$ within the desired frequency interval is compared in Figure \ref{fig1}.
\begin{figure}[!h]
  \centering
  \includegraphics[width=10cm]{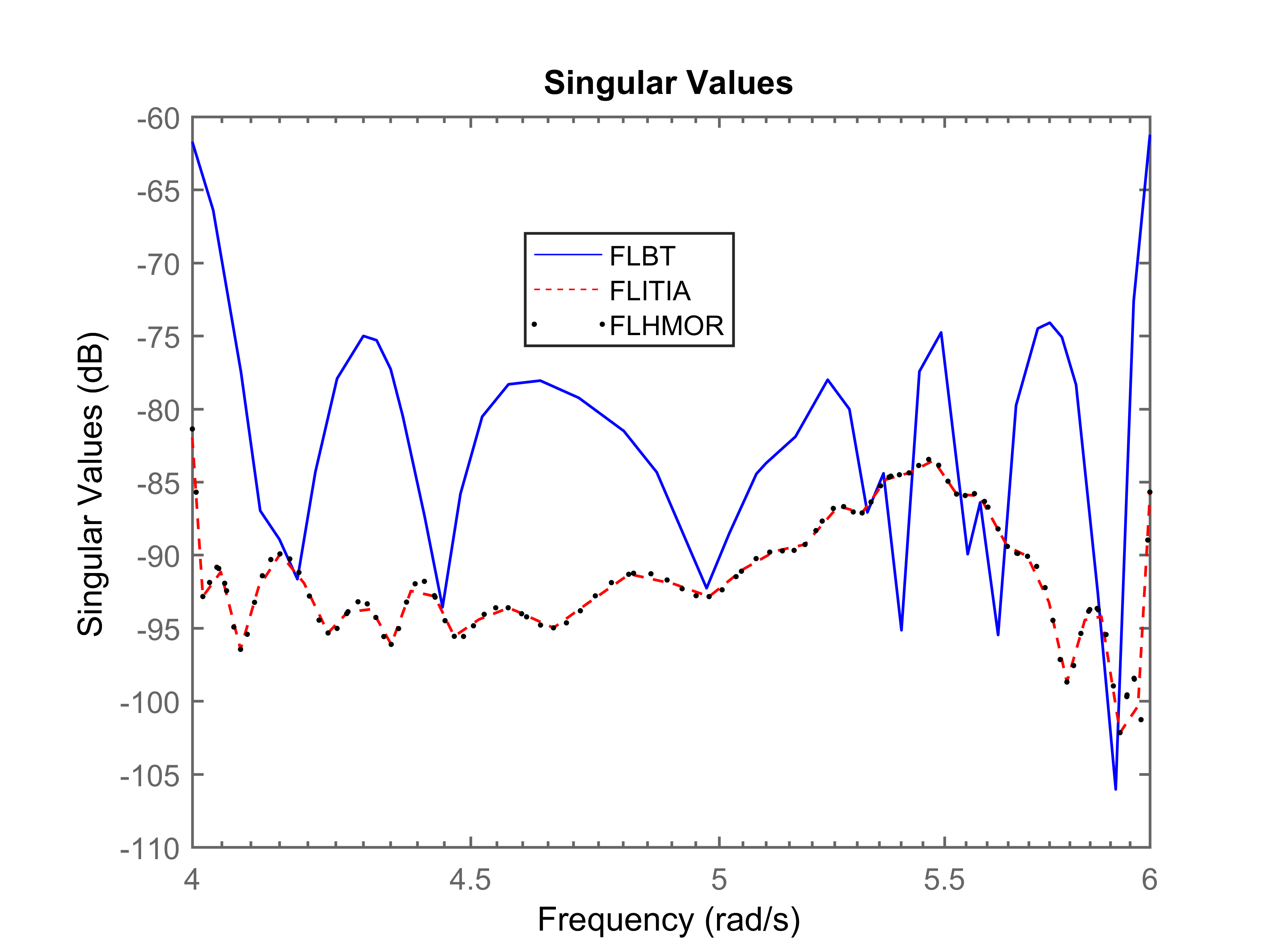}
  \caption{Singular values of $E(s)$ within $[4,6]$ rad/sec}\label{fig1}
\end{figure} It can be seen that FLITIA and FLHMOR offer better accuracy than FLBT.

For the time-limited case, the desired time interval is chosen as $[0,1]$ sec. The clamped beam model is reduced by using TLBT, TLITIA, and TLHMOR, and their accuracy is compared in Table \ref{tab3}. It can be noticed that TLITIA and TLHMOR mostly offer less $||E(s)||_{\mathcal{H}_{2,\tau}}$ than TLBT.
\begin{table}[!h]
\caption{$||E(s)||_{\mathcal{H}_{2,\tau}}$}\label{tab3}
\centering
\begin{tabular}{|c|c|c|c|}
\hline
Order & TLBT & TLITIA & TLHMOR \\ \hline
10    &  0.1637    &    0.1016    &   0.1016     \\ \hline
11    &  0.1200    &    0.1051    &  0.1051      \\ \hline
12    &  0.0872    &     0.0903   &     0.0903   \\ \hline
13    &    0.0662  &     0.0390   &    0.0390    \\ \hline
14    &    0.0594  &    0.0586    &    0.0586    \\ \hline
15    &   0.0018   &    0.0017    &   0.0017     \\ \hline
\end{tabular}
\end{table}
The time-domain error for the ROMs of order $12$ within the desired time interval is compared in Figure \ref{fig2}.
\begin{figure}[!h]
  \centering
  \includegraphics[width=10cm]{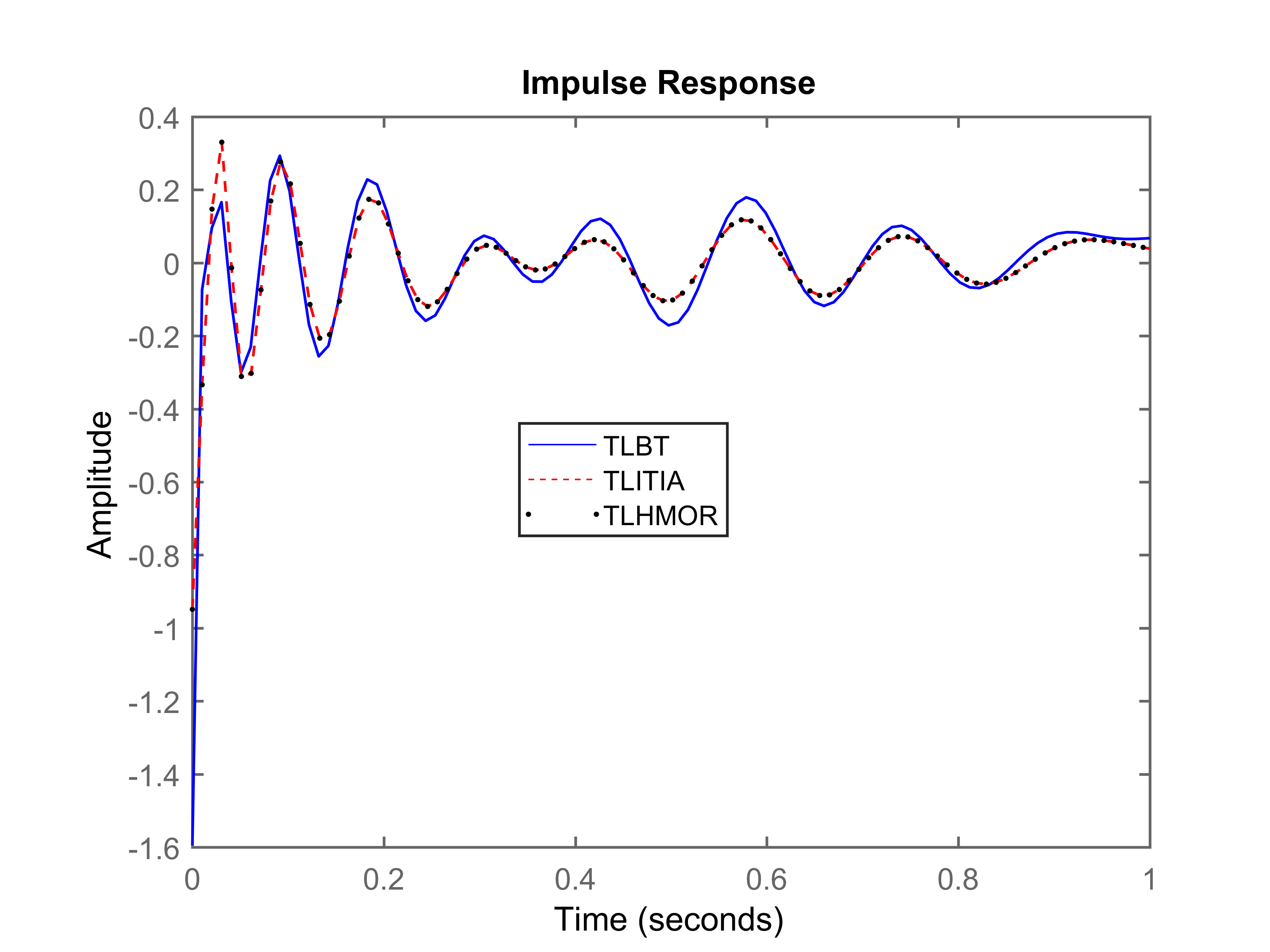}
  \caption{Impulse response of $E(s)$ within $[0,1]$ sec}\label{fig2}
\end{figure}It can be seen that TLITIA and TLHMOR compare well in accuracy with TLBT.
\subsection{Artificial Dynamic System} Let us consider the $1006^{th}$ order artificial dynamic system model as the test system from the benchmark collection of test systems for MOR, cf. \citep{chahlaoui2005benchmark}. For the frequency-limited case, the desired frequency interval is chosen as $[-15,-11]\cup[11,15]$ rad/sec. The artificial dynamic system model is reduced by using FLBT, FLITIA, and FLHMOR, and their accuracy is compared in Table \ref{tab4}. It can be noticed that FLITIA and FLHMOR mostly offer less $||E(s)||_{\mathcal{H}_{2,\omega}}$ than FLBT.
\begin{table}[!h]
\caption{$||E(s)||_{\mathcal{H}_{2,\omega}}$}\label{tab4}
\centering
\begin{tabular}{|c|c|c|c|}
\hline
Order & FLBT & FLITIA & FLHMOR \\ \hline
10    & $2.3514\times 10^{-5}$     &  $3.9357\times 10^{-6}$      & $2.1685\times 10^{-5}$       \\ \hline
11    & $1.5678\times 10^{-5}$     & $3.5242\times 10^{-6}$       & $1.5956\times 10^{-5}$       \\ \hline
12    & $5.7383\times 10^{-5}$     &  $5.8867\times 10^{-6}$      &  $7.2903\times 10^{-6}$      \\ \hline
13    &  $4.2452\times 10^{-5}$    &  $2.4202\times10^{-5}$      &   $6.6805\times 10^{-6}$     \\ \hline
14    &   $3.8084\times 10^{-5}$   &   $9.8140\times 10^{-6}$     &  $7.5714\times 10^{-6}$      \\ \hline
15    &    $5.8612\times 10^{-5}$  &   $1.7097\times 10^{-5}$     &  $1.6830\times 10^{-5}$      \\ \hline
\end{tabular}
\end{table}The frequency-domain error for the ROMs of order $10 $ within the desired frequency interval is compared in Figure \ref{fig3}.
\begin{figure}[!h]
  \centering
  \includegraphics[width=10cm]{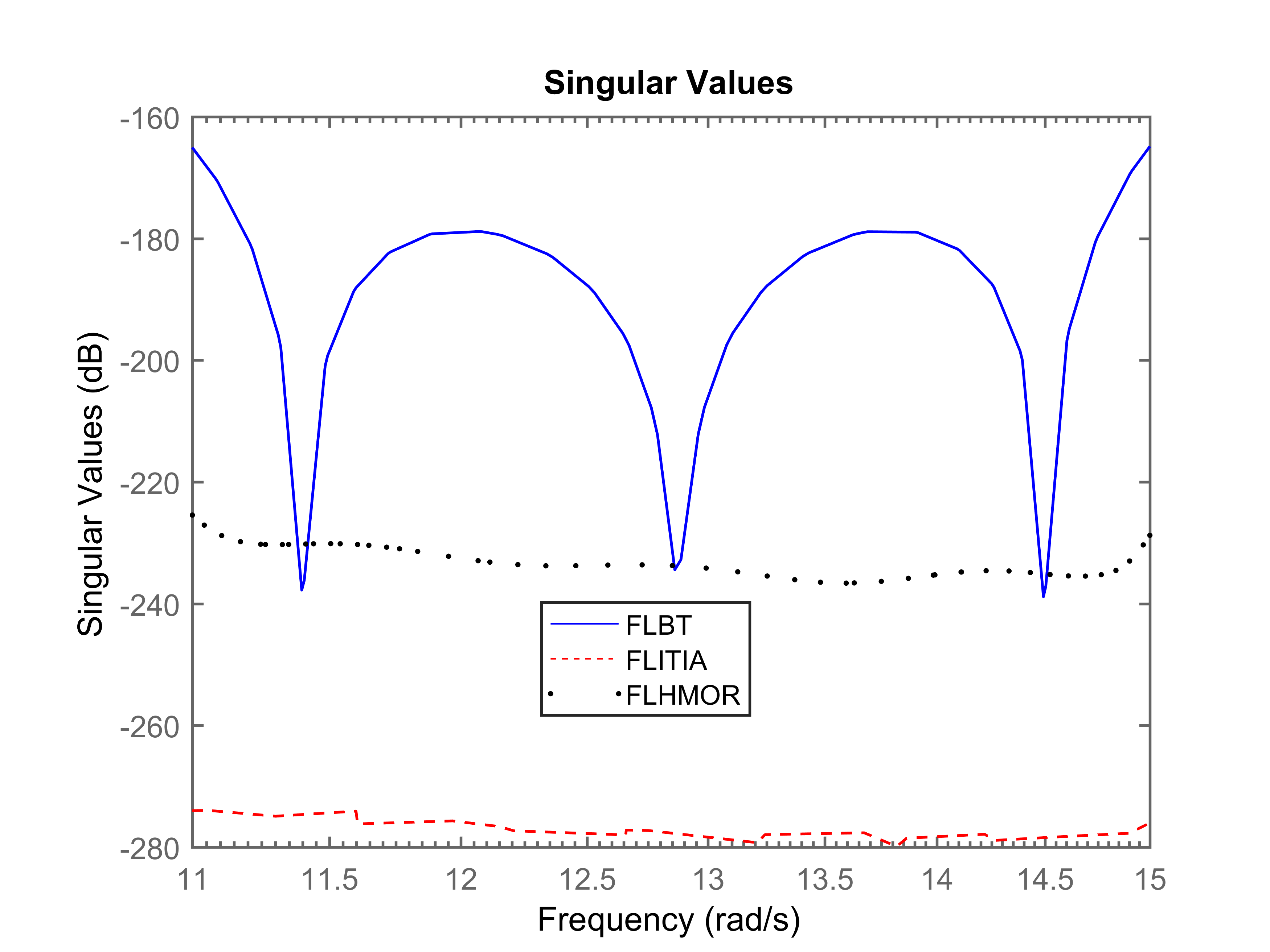}
  \caption{Singular values of $E(s)$ within $[11,15]$ rad/sec}\label{fig3}
\end{figure}It can be seen that FLITIA and FLHMOR offer better accuracy than FLBT.

For the time-limited case, the desired time interval is chosen as $[0,2]$ sec. The artificial dynamic system model is reduced by using TLBT, TLITIA, and TLHMOR, and their accuracy is compared in Table \ref{tab5}. It can be noticed that TLITIA and TLHMOR offer less $||E(s)||_{\mathcal{H}_{2,\tau}}$ than TLBT.
\begin{table}[!h]
\caption{$||E(s)||_{\mathcal{H}_{2,\tau}}$}\label{tab5}
\centering
\begin{tabular}{|c|c|c|c|}
\hline
Order & TLBT & TLITIA & TLHMOR \\ \hline
10    & $0.5170$     &  $0.3428$      &   $0.3428$     \\ \hline
11    &  $0.1562$    &   $0.1030$     &  $0.1030$      \\ \hline
12    &   $0.0460$   &   $0.0312$     &  $0.0312$      \\ \hline
13    &   $0.0131$   &  $0.0093$      &   $0.0093$     \\ \hline
14    &   $0.0036$   &    $0.0026$    &  $0.0026$      \\ \hline
15    &   $9.9176\times 10^{-4}$   &   $7.1795\times 10^{-4}$     &   $7.1861\times 10^{-4}$     \\ \hline
\end{tabular}
\end{table}
The time-domain error for the ROMs of order $15$ within the desired time interval is compared in Figure \ref{fig4}.
\begin{figure}[!h]
  \centering
  \includegraphics[width=10cm]{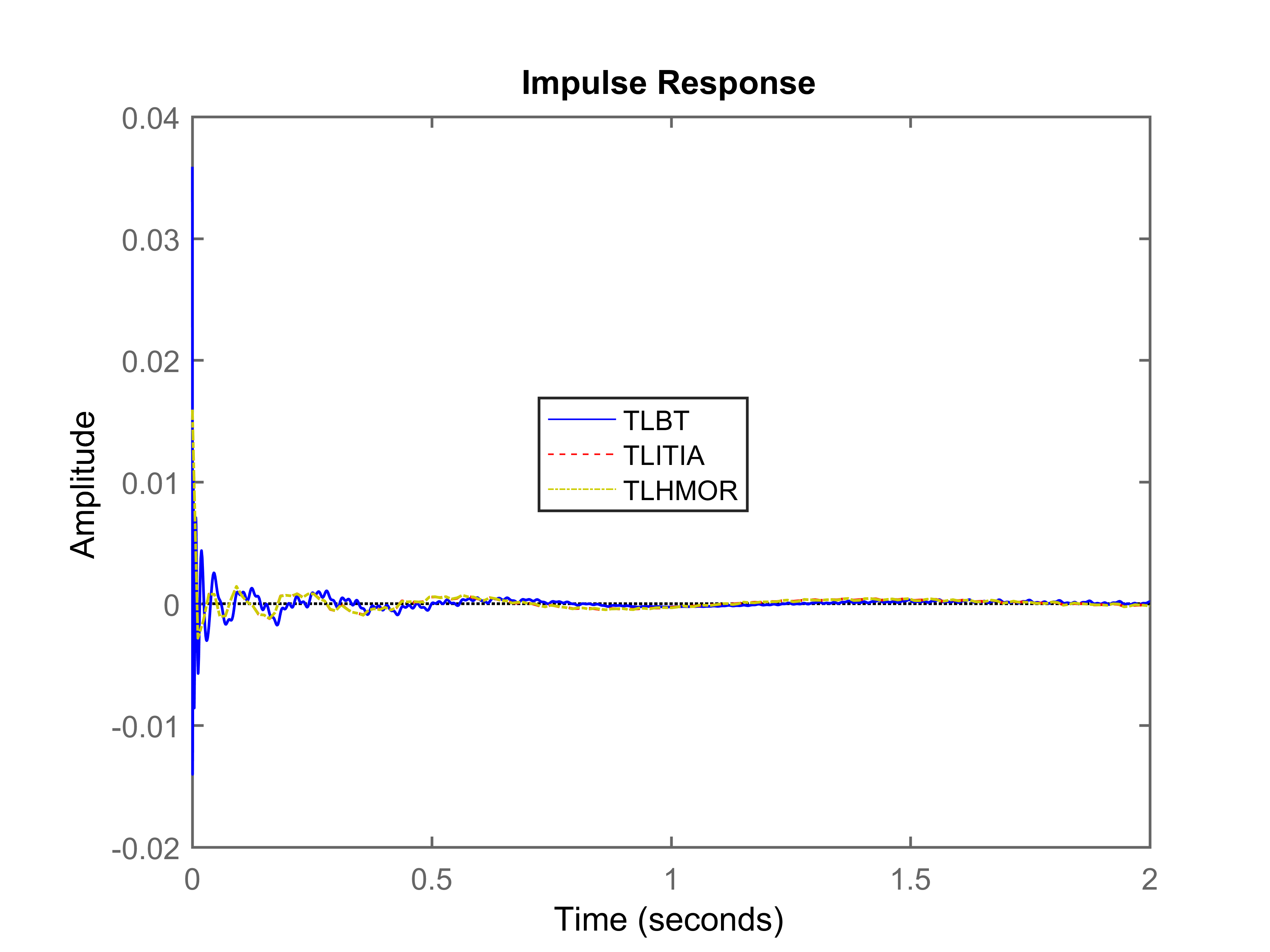}
  \caption{Impulse response of $E(s)$ within $[0,2]$ sec}\label{fig4}
\end{figure}
\subsection{International Space Station} Let us consider the $348^{th}$ order international space station model as the test system from the benchmark collection of test systems for MOR, cf. \citep{chahlaoui2005benchmark}. This model has $3$ inputs and $3$ outputs. For the frequency-limited case, the desired frequency interval is chosen as $[-12,-9]\cup[9,12]$ rad/sec. The international space station model is reduced by using FLBT, FLITIA, and FLHMOR, and their accuracy is compared in Table \ref{tab6}. It can be noticed that FLITIA and FLHMOR offer less $||E(s)||_{\mathcal{H}_{2,\omega}}$ than FLBT.
\begin{table}[!h]
\caption{$||E(s)||_{\mathcal{H}_{2,\omega}}$}\label{tab6}
\centering
\begin{tabular}{|c|c|c|c|}
\hline
Order & FLBT & FLITIA & FLHMOR \\ \hline
15    & $3.4372\times 10^{-5}$     &  $2.4039\times 10^{-5}$      &  $2.4039\times 10^{-5}$      \\ \hline
16    & $2.7377\times 10^{-5}$      &  $1.1905\times 10^{-5}$       &   $1.1905\times 10^{-5}$      \\ \hline
17    &  $5.1045\times 10^{-5}$     &  $1.0804\times 10^{-5}$       &   $1.0804\times 10^{-5}$      \\ \hline
18    & $5.1055\times 10^{-5}$      &   $3.6488\times 10^{-6}$      & $3.6488\times 10^{-6}$         \\ \hline
19    &  $5.0940\times 10^{-5}$     &  $3.4274\times 10^{-6}$       &  $3.4274\times 10^{-6}$       \\ \hline
20    &   $2.8898\times 10^{-5}$    &  $2.9185\times 10^{-6}$       &  $2.9211\times 10^{-6}$       \\ \hline
\end{tabular}
\end{table}There are $9$ frequency-domain plots of the $3\times 3$ error transfer function that are not shown for brevity.
For the time-limited case, the desired time interval is chosen as $[0,2.5]$ sec. The international space model is reduced by using TLBT, TLITIA, and TLHMOR, and their accuracy is compared in Table \ref{tab7}.
\begin{table}[!h]
\caption{$||E(s)||_{\mathcal{H}_{2,\tau}}$}\label{tab7}
\centering
\begin{tabular}{|c|c|c|c|}
\hline
Order & TLBT & TLITIA & TLHMOR \\ \hline
15    & $9.5009\times 10^{-4}$     &  $5.8676\times 10^{-4}$      & $5.8676\times 10^{-4}$       \\ \hline
16    & $6.3547\times 10^{-4}$     & $5.1907\times 10^{-4}$       &  $5.1907\times 10^{-4}$      \\ \hline
17    & $3.8048\times 10^{-4}$     & $3.6969\times 10^{-4}$       & $3.6969\times 10^{-4}$       \\ \hline
18    & $5.6965\times 10^{-4}$     & $4.2228\times 10^{-4}$       & $4.2228\times 10^{-4}$       \\ \hline
19    & $2.5937\times 10^{-4}$     & $1.9977\times 10^{-4}$       & $1.9977\times 10^{-4}$       \\ \hline
20    & $1.8241\times 10^{-4}$     & $1.7952\times 10^{-4}$       & $1.7952\times 10^{-4}$       \\ \hline
\end{tabular}
\end{table}It can be noticed that TLITIA and TLHMOR offer less $||E(s)||_{\mathcal{H}_{2,\tau}}$ than TLBT.
There are $9$ time-domain plots of the $3\times 3$ error transfer function that are not shown for brevity.
\section{Conclusion}
The $\mathcal{H}_{2,\omega}$- and $\mathcal{H}_{2,\tau}$-optimal MOR problems within the oblique projection framework are addressed. It is shown that two out of three optimality conditions can be exactly satisfied within the oblique projection framework, whereas the third optimality condition can be nearly satisfied. The equivalence between gramian-based optimality conditions and the tangential interpolation conditions is also established. Two iterative tangential interpolation algorithms are proposed that nearly satisfy the three optimality conditions upon convergence. The deviations in satisfaction of the optimality conditions decay quickly as the order of the reduced model increases. Further, two stationary point iteration algorithms are proposed that exactly satisfy two optimality conditions upon convergence, while the third optimality condition is nearly satisfied. The deviation in satisfaction of the third optimality condition decays quickly as the order of the reduced model increases. The numerical results confirm that the proposed algorithms construct near-optimal ROMs, which exhibit high fidelity within the desired frequency and time intervals.

\end{document}